%% file: HyperbolicGenerator.tex
\documentclass[a4paper,notitlepage]{article}

\usepackage[utf8]{inputenc}
\usepackage{times}

\usepackage{graphicx}
\usepackage{xspace}

\usepackage{wrapfig}
\usepackage{footnote}
\usepackage{placeins}

\usepackage{subcaption}
\usepackage{amsmath}
\usepackage{amssymb}
\usepackage{amsthm}
\usepackage{fancyhdr}%

\usepackage{enumitem}
\usepackage{pgfplots}
\pgfplotsset{compat=1.9}
\usepackage{hyperref}
\usepackage[ruled,vlined,linesnumbered]{algorithm2e}
\usepackage{microtype}

\usepackage[numbers]{natbib}

\include{abbrv}

\newcommand{\pointset}{\ensuremath{\mathrm{P}}\xspace}
\newcommand{\neighborset}{\ensuremath{\mathrm{N}}}
\newcommand{\candidateset}{\ensuremath{\mathrm{Candidates}}}
\newcommand{\cellset}{\ensuremath{\mathrm{Cells}}}
\newcommand{\overhangset}{\ensuremath{\mathrm{Overhang}}\xspace}

\newcommand{\ringset}{\ensuremath{\varsigma}\xspace}

\definecolor{markedcolor}{RGB}{31,120,180}
\definecolor{plottinggreen}{RGB}{178,223,138}
\definecolor{thirdhue}{RGB}{228,26,28}

\newcommand{\old}[1]{}

\newcommand{\acosh}{\ensuremath{\mathrm{acosh}}}
\newcommand{\asinh}{\ensuremath{\mathrm{asinh}}}

\newcount\shortyear\newcount\shorthour\newcount\shortminute
\shorthour=\time\divide\shorthour by 60\shortyear=\shorthour
\multiply\shortyear by 60\shortminute=\time\advance\shortminute by
-\shortyear
\shortyear=\year\advance\shortyear by -1900

\def\zeit{\number\shorthour:\ifnum\shortminute<10 0\number\shortminute
\else\number\shortminute\fi}

\DeclareMathOperator*{\argmin}{arg\,min}
\newtheorem{theorem}{Theorem}
\newtheorem{lemma}{Lemma}
\newtheorem{proposition}{Proposition}
\newtheorem{corollary}{Corollary}
\title{Updating Dynamic Random Hyperbolic Graphs in Sublinear Time\thanks{
This work is partially supported by German Research Foundation (DFG) grant ME 3619/3-1 (FINCA) within the Priority Programme 1736 \emph{Algorithms for Big Data}.
It was carried out while the authors were affiliated with the Institute of Theoretical Informatics at Karlsruhe Institute of Technology (KIT).
Authors’ present addresses: M. von Looz and H. Meyerhenke, Institute of Computer Science, University of Cologne, Weyertal 121, 50931 Cologne, Germany;
emails: mloozcor@uni-koeln.de, h.meyerhenke@uni-koeln.de.}}
\author{Moritz von Looz \and Henning Meyerhenke}

\begin{document}
\markboth{Moritz von Looz et al.}{Updating Dynamic Random Hyperbolic Graphs in Sublinear Time}
\date{}

\maketitle

\begin{abstract}
Generative network models play an important role in algorithm development, scaling studies, network analysis, and realistic system benchmarks for graph data sets.
A complex network model gaining considerable popularity builds random hyperbolic graphs, generated by distributing points within a disk in the hyperbolic plane and then adding edges between points with a probability depending on their hyperbolic distance.

We present a dynamic extension to model gradual network change, while preserving at each step the point position probabilities.
To process the dynamic changes efficiently, we formalize the concept of a \emph{probabilistic neighborhood}:
Let $P$ be a set of $n$ points in Euclidean or hyperbolic space, $q$ a query point, $\dist$ a distance metric, and $f : \mathbb{R}^+ \rightarrow [0,1]$ a monotonically decreasing function.
Then, the probabilistic neighborhood $N(q, f)$ of $q$ with respect to $f$ is 
a random subset of $P$ and each point $p \in P$ belongs to $N(q,f)$ with probability $f(\dist(p,q))$.
We present a fast, sublinear-time query algorithm to sample probabilistic neighborhoods from planar point sets.
For certain distributions of planar $P$, we prove that our algorithm answers a query in $O((|N(q,f)| + \sqrt{n})\log n)$ time with high probability.
This enables us to process a node movement in random hyperbolic graphs in sublinear time, resulting in a speedup of about one order of magnitude in practice compared to the fastest previous approach.
Apart from that, our query algorithm is also applicable to Euclidean geometry, making it of independent interest for other sampling or probabilistic spreading scenarios.
\end{abstract}

\section{Introduction}
\label{sec:introduction}
Relational data of complex relationships often take the form of \emph{complex networks}, graphs with heterogeneous and often hierarchical structure, low diameter, high clustering, and a heavy-tailed degree distribution~\cite{chakrabarti2006graph}.
Examples include social networks, the graph of hyperlinks between websites, protein interaction networks, and infrastructure routing networks on the autonomous system level~\cite{newman2010networks}.
Frequently found properties in generative models for complex networks are non-negligible clustering (ratio of triangles to triads), a pronounced community structure, and a heavy-tailed degree distribution~\cite{chakrabarti2006graph}, such as a power-law.
Moreover, complex networks often have the \emph{small-world property}, \ie the distance between any two vertices is surprisingly small, often even regardless of network size. 

\subsection{Random Hyperbolic Graphs}
Many properties of networks can be explained with a hidden underlying geometry~\cite{k-gruhs-07}.
While Euclidean geometry is a natural fit for mesh or unit-disk networks, Krioukov et al.~\cite{Krioukov2010} show that properties of complex networks follow naturally from \emph{hyperbolic} geometry,
which has negative curvature and is the basis for one of the three isotropic spaces.  (The other two are Euclidean (flat) and spherical geometry (positive curvature).)
Their proposed family of \emph{random hyperbolic graphs} (RHG) thus gained considerable interest as a generative network model~\cite{bode2014probability,DBLP:conf/icalp/GugelmannPP12,raey}.
These graphs offer a power-law degree distribution with an adjustable exponent and a hierarchical community structure. In addition, the clustering coefficient is stable with increasing graph size, in contrast to, for example, the popular R-MAT model~\cite{KoPiPlSe14,DBLP:conf/icalp/GugelmannPP12,2015arXiv150103545V}.

To sample a RHG, vertices are distributed randomly on a hyperbolic disk of given radius $R$ and each pair of vertices is connected by an edge with a probability that decreases with the hyperbolic distance between them.
\citet{papadopoulos2012popularity} argue that, when given in polar coordinates, the angular and radial coordinates model popularity and similarity of vertices in a network.
Due to the hyperbolic geometry, vertices close to the center of the disk have a higher degree than vertices on the outer rim of the hyperbolic disk.
If two nodes have similar angular coordinates, they are more likely to be connected and have similar neighborhoods.

For the static model of RHGs, \citet{bringmann2015geometric}. describe an algorithm to sample graphs in expected linear time; it uses a cell data structure and iterates over pairs of cells to generate the edges.
This algorithm and the implementation of \citet{blsius_et_al:LIPIcs:2016:6367} are an improvement over our previous work, which generated static RHGs with $n$ vertices and $m$ edges in $O((n^{3/2}+m)\log n)$~\cite{Looz2015HRG}.
Our previous work used a polar quadtree on the Poincaré disk model of hyperbolic space to provide range queries for a restricted model of RHGs, an idea we extend in this paper.

Many networks are not static but change over time, however, thus giving rise to the field of \emph{dynamic network analysis} and the need for dynamic generative models.
Similar to static generative models, they fulfill an analogous role in the development, benchmarking, and scaling studies of dynamic graph algorithms.
\citet{Papadopoulos2010} consider a dynamic model for RHGs in which vertices are deleted and reinserted at random positions to model sudden site outages and additions in an infrastructure network.
In many complex networks, e.\ g.,\ social networks, changes happen more gradually, though.

We thus introduce a dynamic model with gradual node movement, which reflects for example vertices becoming similar and their neighborhoods merging, or vertices waxing and waning in popularity.
The position probability distribution of points is kept; each snapshot of the dynamic model is a random hyperbolic graph with the usual properties.

\subsection{Probabilistic Neighborhood Queries}
While efficient sampling algorithms for static RHGs exist~\cite{bringmann2015geometric}, resampling the neighborhood of a moved node in a dynamic model required linear time so far.

This task of sampling a neighborhood whose elements are probabilistic not only occurs in dynamic random hyperbolic graphs.
Connection probabilities depending on the distance frequently happen in Euclidean applications as well:
The probability that a customer shops at a certain physical store shrinks with increasing distance to it.
In disease simulations, if the social interaction graph is unknown but locations are available, disease transmission 
can be modeled as a random process with infection risk decreasing with distance.
Moreover, the wireless connections between units in an ad-hoc network are fragile and collapse more frequently with higher distance.

To generalize these scenarios, we define the notion of a \emph{probabilistic neighborhood} in spatial data sets, both Euclidean and hyperbolic:
Let a set $P$ of $n$ points in $\mathbb{R}^d$, a query point $q \in \mathbb{R}^d$, a distance metric $\dist$, 
and a monotonically decreasing function $f : \mathbb{R}^+ \rightarrow [0,1]$ be given.
Then, the probabilistic neighborhood $N(q, f)$ of $q$ with respect to $f$ is 
a random subset of $P$ and each point $p \in P$ belongs to $N(q,f)$ with probability $f(\dist(p,q))$.
A straightforward query algorithm for sampling a probabilistic neighborhood
would iterate over each point $p \in P$ and sample for each whether it is included in $\neighborset(q, f)$.
This has a running time of $\Theta(n \cdot d)$ per query point, which is prohibitive for repeated queries in large data sets.
Thus we are interested in a faster algorithm for such a \emph{probabilistic neighborhood query} (PNQ, spoken as ``pink'').
We restrict ourselves to the planar case in this work, but the algorithmic principle is generalizable to higher dimensions.

Since the neighborhood of a vertex in a random hyperbolic graph is an instance of such a probabilistic neighborhood,
we can use a fast PNQ query method to support a faster dynamic generative model.

\subsection{Outline and Contribution}
After introducing notation and related work (both Section~\ref{sec:preliminaries}), we describe a dynamic model for random hyperbolic graphs with gradual change (Section~\ref{sec:dynamic-model}).
We then develop, analyze, implement, and evaluate a quadtree-based index structure (Section~\ref{sub:qt-augment}) 
and query algorithm (Sections~\ref{sub:baseline-algo} and~\ref{sec:subtree-aggr}) that 
together provide sublinear probabilistic neighborhood queries in the Euclidean and hyperbolic plane.
These can be used to generate updates to random hyperbolic graphs as described above,
but can be of independent interest for completely different application areas
as well (for a simple disease simulation, cf.\ the conference version~\cite{vonLooz2016}).

With some assumptions (which are fulfilled in random hyperbolic graph generation), we show a time complexity of $O((|\neighborset(q, f)| + \sqrt{n})\log n)$ with high probability (whp)\footnote{\emph{With high probability} denotes a probability $\geq 1-1/n$ for $n$ sufficiently large.}
to find a probabilistic neighborhood $\neighborset(q,f)$ among $n$ points (Section~\ref{subsec:subtree-aggregation-complexity}).
We implement the dynamic updates in both our data structure and the data structure of \citet{blsius_et_al:LIPIcs:2016:6367}.
On our data structure, processing a node movement is up to two orders of magnitude faster, in the order of milliseconds for graphs with hundreds of milllions of vertices.

Both the generator code and the network analysis modules are available in version 4.1 of the toolkit NetworKit~\cite{networkit2016journal}, which is aimed at large-scale network analysis.

\section{Preliminaries}
\label{sec:preliminaries}

\subsection{Notation}
\label{sub:notation}
\newcommand{\hyperbolic}{\ensuremath{\mathbb{H}}}
\newcommand{\Euclidean}{\ensuremath{\mathbb{E}}}
We use the usual graph notation of a graph $G$ consisting of a vertex set $V$ consisting of $n$ points and an edge set $E\subseteq V\times V$ with $m$ edges. 

For the probabilistic neighborhood queries, let the input be given as set $\pointset$ of $n$ points. 
The points in $\pointset$ are distributed in a disk $\mathbb{D}_R$ of radius $R$ in the hyperbolic or Euclidean plane, the distribution is given by a probability density function $j(\phi, r)$ for an angle $\phi$ and a radius $r$.
For our theoretical results to hold, we require $j$ to be known, continuous and integrable.
Furthermore, $j$ needs to be rotationally invariant -- meaning that $j(\phi_1, r) = j(\phi_2, r)$ for any radius $r$ and any two angles $\phi_1$ and $\phi_2$ -- and positive within  $\mathbb{D}_R$,
so that $j(r) > 0 \Leftrightarrow r < R$.
Due to the rotational invariance, $j(\phi, r)$ is the same for every $\phi$ and we can write $j(r)$.
Likewise, we define $J(r)$ as the indefinite integral of $j(r)$ and normalize it so that $J(R) = 1$ (also implying $J(0) = 0$). The value $J(r)$ then gives the fraction of probability mass inside radius $r$.

For the distance between two points $p_1$ and $p_2$, we use $\dist_\hyperbolic{}(p_1, p_2)$ for the hyperbolic and $\dist_\Euclidean{}(p_1, p_2)$ for the Euclidean case.
We may omit the index if a distinction is unnecessary.
As mentioned, a point $p$ is in the probabilistic neighborhood of query point $q$ with probability $f(\dist(p, q))$.
Thus, a \emph{query pair} consists of a query point $q$ and a function $f : \mathbb{R}^+ \rightarrow [0,1]$ that maps distances to probabilities.
The function $f$ needs to be monotonically decreasing but may be discontinuous, requirements that are fulfilled for the application of random hyperbolic graphs.
(Note that $f$ can be defined differently for each query.
This might be useful when applying PNQs to spatial data sets, where after one preprocessing step queries of different types can be handled without changing the data structure.)
The query result, the probabilistic neighborhood of $q$ \wrt $f$, is denoted by the set $\neighborset(q,f) \subseteq P$.
For the algorithm analysis, we use two additional sets for each query $(q,f)$:
\begin{itemize}
 \item $\candidateset(q,f)$: neighbor candidates examined when executing such a query,
 \item $\cellset(q, f)$: quadtree cells examined during execution of the query.
\end{itemize}
Note that the sets $\neighborset(q,f), \candidateset(q,f)$ and $\cellset(q,f)$ are probabilistic, thus theoretical results about their size are usually only with high probability.

\subsection{Related Range Queries}
\label{sec:related-range-queries}

Since PNQs are applicable to other problems in for example Euclidean spatial databases, we discuss related query algorithms and data structures.

\paragraph{Fast Deterministic Range Queries}
Numerous index structures for fast range queries on spatial data exist.
Many such index structures are based on trees or variations 
thereof, see Samet's book~\cite{Samet:2005:FMM:1076819} for a comprehensive overview.
I/O efficient worst case analysis is usually performed using the EM model,
see \eg~\cite{Arge:2012:ISD:2367574.2367575}. In more applied settings, average-case performance is of
higher importance, which popularized R-trees or newer variants thereof, \eg~\cite{Kamel:1994:HRI:645920.673001}.
Concerning (balanced) quadtrees and kd-trees for spatial dimension $d$, it is known that queries require $O(d \cdot n^{1-1/d})$ time
(thus $O(\sqrt{n})$ in the planar case)~\cite[Ch.~1.4]{Samet:2005:FMM:1076819}.
Regarding PNQs our algorithm matches this query complexity up to a logarithmic factor.
Yet note that, since for general $f$ and $\dist$ in our scenario all points in the set $P$ could be neighbors, 
data structures for deterministic queries cannot solve a PNQ efficiently without adaptations.

\citet{Hu2014independent} give a query sampling algorithm for one-dimensional data that,
given a set $P$ of n points in $\mathbb{R}$,  an interval $q = [x,y]$ and an integer, $t \geq 1$, returns $t$ elements uniformly sampled from $P \cap q$.
They describe a structure of $O(n)$ space that answers a query in $O(\log n + t)$ time and supports updates in $O(\log n)$ time.
While also offering query sampling, PNQs differ from the problem considered by Hu et al. in two aspects: We consider  not only the $1$-dimensional case, and our sampling probabilities (user-defined with a distance-dependent function) are not necessarily uniform.

\paragraph{Range Queries on Uncertain Data}
During the previous decade probabilistic queries \emph{different} from PNQs have become popular.
The main scenarios can be put into two categories~\cite{pei2008query}: (i) Probabilistic databases contain entries
that come with a specified confidence (\eg sensor data whose accuracy is uncertain) and
(ii) objects with an uncertain location, \ie the location is specified by a probability distribution.
Both scenarios differ under typical and reasonable assumptions from ours: 
Queries for uncertain data are usually formulated to return \emph{all} points in the neighborhood
whose confidence/probability exceeds a certain threshold~\cite{kriegel2007probabilistic},
or computing points that are possibly nearest neighbors~\cite{agarwal2013nearest}.

In our model, in turn, the choice of inclusion of a point $p$ is a random choice for every different $p$. In particular, 
depending on the probability distribution, \emph{all}
nodes in the plane can have positive probability to be part of some other's neighborhood.
In the related scenarios this would only be true with extremely small confidence values or extremely
large query circles.

\subsection{Graphs in Hyperbolic Geometry}
\label{sub:hyperbolic-introduction}
\citet{Krioukov2010} relate complex networks with hierarchical structures to hyperbolic geometry and introduce the family of random hyperbolic graphs, which develop a power-law degree distribution, high clustering and other properties of complex networks simply from their geometry.
Numerous other generative graph models, including ones for complex networks and not based on geometry, exist. For a short overview
cf.\ \cite{Staudt2017}.
All these models cover different aspects of network formation and the graphs generated by them have systematically different properties.
No model is widely accepted as covering the majority of use cases.

In the RHG model by~\citet{Krioukov2010}, vertices are generated as points in polar coordinates $(\phi, r)$ on a disk of radius $R$ in the hyperbolic plane with curvature $-\zeta^2$.
We denote this disk with $\mathbb{D}_R$.
The angular coordinate $\phi$ is drawn from a uniform distribution over $[0,2\pi]$, while the probability density for the radial coordinate $r$ is given by~\citet[Eq.~(17)]{Krioukov2010} and controlled by a dispersion parameter $\alpha$:
\begin{equation}
 f(r) = \alpha\frac{\sinh(\alpha r)}{\cosh(\alpha R)-1}
 \label{eq:base-radial-distribution}
\end{equation}
For $\alpha=1$, this yields a uniform distribution on the hyperbolic plane within $\mathbb{D}_R$.
For lower values of $\alpha$, vertices are more likely to be in the center, for higher values more likely at the border of $\mathbb{D}_R$.

We denote the hyperbolic distance between two points $p_1$ and $p_2$ with $\mathrm{dist}_{\mathbb{H}}(p_1,p_2)$.
In the model, any two vertices $u$ and $v$ are connected by an edge with a probability depending on their distance, given by Eq.~\eqref{eq:Krioukov-equation} and parametrized by a temperature $T$.
\begin{equation}
 p(\{u,v\}\in E) = (1+e^{(1/T)\cdot(\mathrm{dist}_{\mathbb{H}}(u,v)-R)/2})^{-1}
 \label{eq:Krioukov-equation}
\end{equation}
For the limiting case of $T=0$, the neighborhood of a point consists of exactly those points within a hyperbolic circle of radius $R$, giving rise to the name \emph{threshold random hyperbolic graphs}.
In the other extreme of $T=\infty$, the geometry's influence vanishes and the resulting model resembles the Erdos-Reyni-model with a binomial degree distribution.
In this work, we consider the general case of $0 \leq T < \infty$ unless otherwise noted.

Several works have analyzed the properties of the resulting graphs theoretically.
\citet[Eq.~(29)]{Krioukov2010} show that for $\alpha/\zeta \geq \frac{1}{2}$, the degree distribution follows a power law with exponent $\gamma := 2\cdot \alpha/\zeta +1$.
\citet{DBLP:conf/icalp/GugelmannPP12} prove non-vanishing clustering and a low variation of the clustering coefficient.
\citet{raey} discuss the size of the giant component and the probability that the graph is connected~\cite{bode2014probability}.
They also show~\cite{bode2014probability} that the curvature parameter $\zeta$ can be fixed while retaining all degrees of freedom, leading us to assume $\zeta=1$ from now on without loss of generality.
For $2 < \gamma < 3$, \citet{kiwi2015bound} bound the diameter asymptotically almost surely to $O((\log n)^{32/((5-\gamma)(3-\gamma))})$ and \citet{Friedrich2015} improve that bound to $O((\log n)^{2/(3-\gamma)})$.
The average degree $\overline{k}$ of a random hyperbolic graph is controlled with the radius $R$, using an approximation given by~\citet[Eq.~(22)]{Krioukov2010}.
This radius is commonly set to $R = 2 \log n + C$ with a user-defined constant $C$, leading to a stable average degree with changing graph size.
\citet{Krioukov2010} choose a polar representation of the hyperbolic plane in which the radial coordinate $r_\mathbb{H}$ of a point $p_\mathbb{H} = (\phi_\mathbb{H}, r_\mathbb{H})$ is set to the hyperbolic distance to the origin: \( r_\mathbb{H} = \mathrm{dist}_{\mathcal{\mathbb{H}}}(p_\mathbb{H},(0,0))\).
They call this representation the \emph{native representation}.
An example graph with 500 vertices, $R\approx 5.08$, $T=0$ and $\alpha=0.8$ in this representation is shown in Figure~\ref{fig:hyperbolic-graph-native}. 
For the purpose of illustration in the figure, we choose a vertex $u$ (the bold blue vertex) and an artificially small example neighborhood, adding edges $(u,v)$ for all vertices $v$ where $\mathrm{dist}_{\mathbb{H}}(u,v) \leq 0.2\cdot R$\footnote{Depicting neighborhoods as in the actual model of RHGs would result in graph too dense to be useful in a visualization.}.
The neighborhood of $u$ then consists of vertices within a hyperbolic circle (marked in blue).
\begin{figure}
\centering
\includegraphics[width=0.5\linewidth]{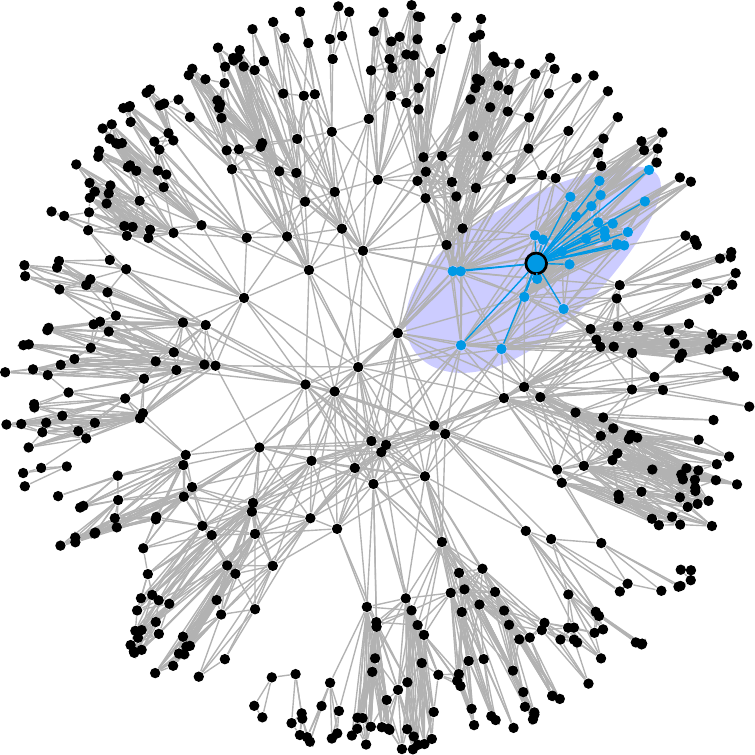}
\caption{Graph in hyperbolic geometry with unit-disk neighborhood. Neighbors of the bold blue vertex are in the hyperbolic circle, marked in blue.
In this visualization, an edge $(u,v)$ is only added if $\dist_\hyperbolic{}(u,v)\leq 0.2R$.}
\label{fig:hyperbolic-graph-native} 
\end{figure}

\subsection{Fast Graph Generation}
As discussed in Section~\ref{sec:introduction}, one application for PNQs are dynamic updates for random hyperbolic graphs.

In previous work we designed a static 
generator~\cite{Looz2015HRG} for a restricted model; a query also runs in $O((n^{3/2}+\neighborset)\log n)$ whp, leading to a time of $O((n^{3/2}+m) \log n)$ whp for the whole graph with $m$ edges.
The range queries discussed there are facilitated by a polar quadtree which supports only deterministic queries. Consequently, the queries result in unit-disk graphs in the hyperbolic plane 
and can be considered as a special case of the current work (a step function $f$ with values 0 and 1 results in a deterministic query).

Our major technical inspiration for enhancing the quadtree for probabilistic neighborhoods is the work of \citet{batagelj2005efficient}. They were the first to present a random sampling 
method to generate Erd\H{o}s-R\'{e}nyi-graphs with $n$ nodes and $m$ edges in $O(n + m)$ time complexity.
Faced with a similar problem of selecting each of $n^2$ possible edges with a constant probability $p$, they designed an 
efficient algorithm with the following idea:
since the gaps between independently selected elements follow a geometric distribution, it is possible to simply sample the gaps from the geometric distribution directly and skip those elements that are not selected while attaining the same probabilities.

\subsection{Quadtree Specifics}
\label{sub:prelim-quadtree}
Our key data structure for fast queries is a region quadtree in the Euclidean or hyperbolic plane.\footnote{Driven by the requirements of random hyperbolic graphs, we use a polar quadtree.}
While they are less suited to higher dimensions as for example k-d-trees, the complexity is comparable in the plane.
For the (circular) range queries we discuss, quadtrees have the significant advantage of a bounded aspect ratio:
A cell in a $k$-d-tree might extend arbitrarily far in one direction, rendering theoretical guarantees about the area affected by the query circle difficult.
In contrast, the region covered by a quadtree cell is determined by its position and level.
We mostly reuse our previous definition~\cite{Looz2015HRG} of the quadtree:
A node in the quadtree is defined as a tuple $(\mathrm{min}_\phi, \mathrm{max}_\phi, \mathrm{min}_r, \mathrm{max}_r)$
with \(\mathrm{min}_\phi \leq \mathrm{max}_\phi\) and \(\mathrm{min}_r \leq \mathrm{max}_r\).
It is responsible for a point $p = (\phi_p, r_p)$ exactly if
$(\mathrm{min}_\phi \leq \phi_p < \mathrm{max}_\phi)$ and $(\mathrm{min}_r \leq r_p < \mathrm{max}_r)$.
We call the region represented by a particular quadtree node its quadtree \emph{cell}.
The quadtree is parametrized by its radius $R$, the $\mathrm{max}_r$ of the root cell.
If the probability distribution $j$ is known (which we assume for our theoretical results), 
we set the radius $R$ to $\argmin_r J(r) = 1$, \ie to the minimum radius that contains the full probability mass.
If only the points are known, the radius is set to include all of them.



\section{Dynamic Model}
\label{sec:dynamic-model}
\citet{Papadopoulos2010} examine greedy routing in random hyperbolic graphs and for this purpose propose a dynamic model in which nodes join the network at the border of a growing disk and depart randomly.
While this is a suitable dynamic behavior for modeling internet infrastructure with sudden site failures or additions,
change in e.\ g.,\ social networks happens more gradually; people hardly leave society completely and rejoin it at a random position.

To model such a gradual change in networks, we design and implement a dynamic version with node movement.
Such a model should fulfill several objectives:
First, it should be \emph{consistent}: After moving a node, the network may change, but properties should stay the same \emph{in expectation}.
Since the properties emerge from the node positions, the probability distribution of node positions needs to be preserved.
Second, the movement should be \emph{directed}: If the movement direction of a node at time $t$ is completely independent from the direction at $t+1$,
the result would be a simulated Brownian motion with the same links vanishing and reappearing repeatedly.

In our implementation, movement happens in discrete time steps.
We attain the first objective by scaling the movements along the radial axis, as given by Theorem~\ref{thm:dynamic-consistency}.
The second objective is fulfilled by initially setting step values $\tau_\phi$ and $\tau_r$ for each node and using them in each movement step.
As a result, if a node $i$ moves in a certain direction at time $t$, it will move in the same direction at $t+1$,
except if the new position would be outside the hyperbolic disk $\mathcal{D}_R$.
In this case, the movement is inverted and the node ``bounces'' off the boundary.
The different probability densities in the center of the disk and the outer regions are translated into movement speed:
A node is less likely to be in the center; thus it needs to spend less time there while traversing it, resulting in a higher speed.

We implement this movement in two phases:
In the initialization, step values $\tau_\phi$ and $\tau_r$ are assigned to each node according to the desired movement.
Each movement step of a node then consists of a rotation and a radial movement.
This step is described in Algorithm~\ref{algo:movement}; a visualization of the radial movement is shown in Figure~\ref{fig:visualization-dynamic}.
\begin{algorithm}[h]
\KwIn{$\phi, r, \tau_\phi, \tau_r, R, \alpha.$}
\KwOut{$\phi_{\text{new}} r_{\text{new}}$}
\begin{enumerate}
\item x = $\sinh(r\cdot\alpha)$\;
\item y = x+$\tau_r$\;
\item z = $\asinh(y)/\alpha$\;
\item $\phi_{\text{new}} = (\phi+\tau_\phi) \bmod 2\pi$
\item \textbf{Return} ($\phi_{\text{new}}$, z)
\end{enumerate}
\caption{move($\phi, r$) - Movement step in dynamic model}
\label{algo:movement}
\end{algorithm}

\begin{figure}[tb]
\centering
\includegraphics[width=0.5\linewidth]{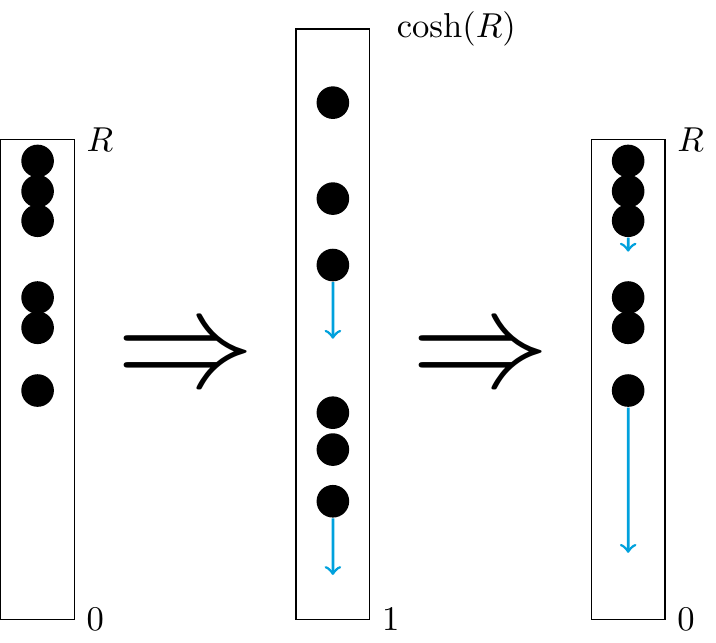}
\caption{For each movement step, radial coordinates are mapped into the interval $[1,\sinh(\alpha R))$, where the coordinate distribution is uniform.
Adding $\tau_r$ and transforming the coordinates back results in correctly scaled movements.}
\label{fig:visualization-dynamic}
\end{figure}

If the new node position would be outside the boundary ($r > R$) or below the origin ($r < 0$), the movement is reflected and $\tau_r$ set to $-\tau_r$.
\begin{theorem}
Let $f_{r,\phi}((p_r, p_\phi))$ be the probability density of point positions, given in polar coordinates.
Let $\mathrm{move}((p_r, p_\phi))$ (Algorithm~\ref{algo:movement}) be a movement step.
Then, the node movement preserves the distribution of angular and radial distributions:
$f_{r,\phi}(\mathrm{move}((p_r, p_\phi))) = f_{r,\phi}((p_r, p_\phi))$.
\label{thm:dynamic-consistency}
\end{theorem}
\begin{proof}
Since the distributions of angular and radial coordinates are independent, we consider them separately: $f_{r,\phi} (p_r, p_\phi) = f_r (p_r) \cdot f_\phi (p_\phi)$.

As introduced in Eq.~(\ref{eq:base-radial-distribution}), the radial coordinate $r$ is sampled from a distribution with density $\alpha\sinh(\alpha r) / (\cosh(\alpha R) -1)$.
We introduce random variables $X, Y, Z$ for each step concerning the radial coordinates in Algorithm~\ref{algo:movement}, each is denoted with the upper case letter of its equivalent. 
An additional random variable $Q$ denotes the pre-movement radial coordinate.
The other variables are defined as $X = \sinh(Q\cdot\alpha)$, $Y = X+\tau_r$ and $Z=\asinh(Y)/\alpha$. 

Let $f_Q, f_X, f_Y$ and $f_Z$ denote the density functions of these variables:
\begin{align}
  f_Q(r) &= \frac{\alpha\sinh(\alpha r)}{\cosh(\alpha R) -1}\\
  f_X(r) &= f_Q\left(\frac{\asinh(r)}{\alpha}\right) = \frac{\alpha r}{\cosh(\alpha R) -1}\\
  f_Y(r) &= f_X(r-\tau_r) = \frac{\alpha r-\tau_r}{\cosh(\alpha R) -1}\\
  f_Z(r) &= f_Y(\sinh(r\cdot\alpha)) = \frac{\alpha \sinh(\alpha r) -\tau_r}{\cosh(\alpha R) -1}  &= f_Q(r) - \frac{\tau_r}{\cosh(\alpha R) -1}
\end{align}

The distributions of $Q$ and $Z$ only differ in the constant addition of $\tau_r / (\cosh(\alpha R)-1)$.
Every $(\cosh(\alpha R)-1) / \tau_r$ steps, the radial movement reaches a limit (0 or $R$) and is reflected, causing $\tau_r$ to be multiplied with -1.
On average, $\tau_r$ is thus zero and $F_Q(r)$ = $F_Z(r)$.

A similar argument works for the rotational step: 
While the rotational direction is unchanged, the change in coordinates is balanced by the addition or subtraction of $2\pi$ whenever the interval $[0,2\pi)$ is left, leading to an average of zero in terms of change.
\end{proof}

Experiments supporting the consistency of properties after node movement can be found in the extended preliminary version of \citet{DBLP:conf/hpec/LoozOLM16}.

\section{Baseline Query Algorithm}
\label{sec:baseline}
After moving a vertex, its edges need to be resampled to reflect its new position.
As discussed before, an algorithm to sample probabilistic neighborhood queries can be used to process such an update.
After detailing the construction of the quadtree data structure for this purpose (Section~\ref{sub:qt-augment}) and several theoretical results, we describe a baseline version of such a query algorithm (Section~\ref{sub:baseline-algo}).
This algorithm introduces the main idea, but is asymptotically not faster than the 
straightforward approach of probing every distance and throwing a biased coin. In Section~\ref{sec:subtree-aggr}, the query algorithm is refined to support faster queries.

\subsection{Quadtree Construction}
\label{sub:qt-augment}
%
\begin{figure}
 \centering
 \includegraphics[width=0.5\linewidth]{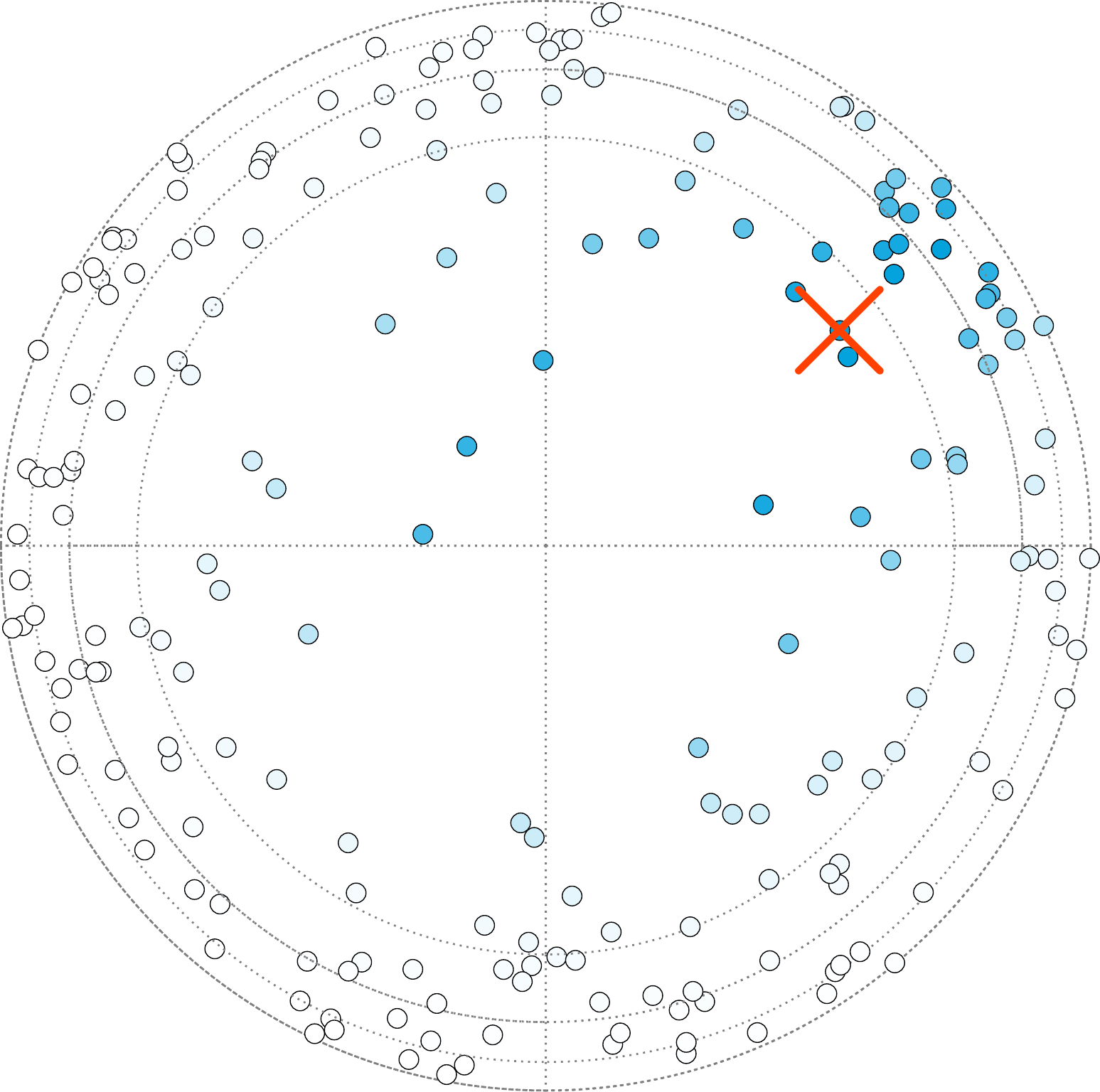}
 \caption{Query over 200 points in a polar hyperbolic quadtree, with $f(d) := 1/(e^{(d-7.78)}+1)$ and the query point $q$ marked by a red cross.
Points are colored according to the probability that they are included in the result. Blue represents a high probability, white a probability of zero.}
\label{fig:visualization-point-probabilities}
\end{figure}


At each quadtree node $v$, we store the size of the subtree rooted there.
We then generalize the rule for node splitting to handle point distributions $j$ as defined
in Section~\ref{sub:notation}:
As is usual for quadtrees, a leaf cell $c$ is split into four children when it exceeds its fixed capacity.
Since our quadtree is polar, this split happens once in the angular and once in the radial direction.
Due to the rotational symmetry of $j$, splitting in the angular direction is straightforward as the angle range is halved: $\mathrm{mid}_\phi := \frac{\mathrm{max}_\phi+\mathrm{min}_\phi}{2}$.
For the radial direction, we choose the splitting radius to result in an equal division of probability mass.
The total probability mass in a ring delimited by $\min_r$ and $\max_r$ is $J(\mathrm{max}_r) - J(\mathrm{min}_r)$.
Since $j(r)$ is positive for $r$ between $R$ and 0, the restricted function $J|_{[0,R]}$ defined above is a bijection.
The inverse $(J|_{[0,R]})^{-1}$ thus exists and we set the splitting radius $\mathrm{mid}_r$ to $(J|_{[0,R]})^{-1}\left(\frac{J(\mathrm{max}_r) + J(\mathrm{min}_r)}{2}\right)$.

Figure~\ref{fig:visualization-point-probabilities} visualizes a point distribution on a hyperbolic disk with 200 points
and Figure~\ref{fig:quadtree-example} its corresponding quadtree.

\begin{figure}[b]
 \includegraphics[width=\linewidth]{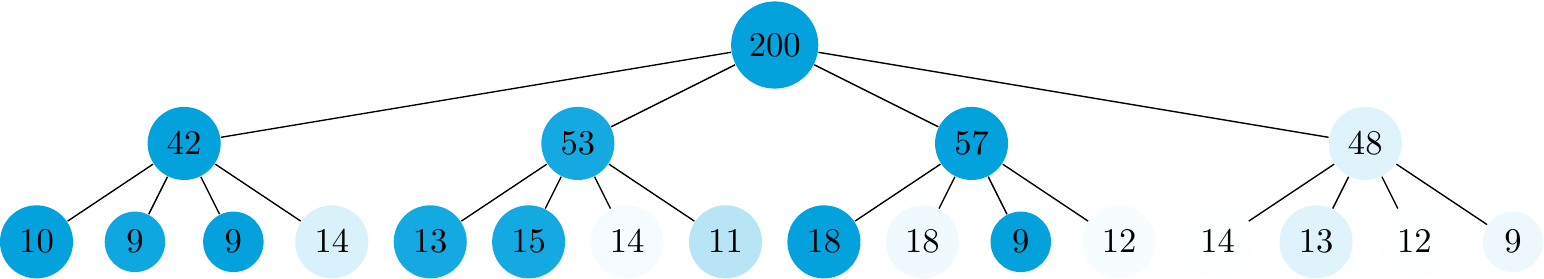}
 \caption{Visualization of the data structure used in Figure~\ref{fig:visualization-point-probabilities}.
 Quadtree nodes are colored according to the upper probability bound for points contained in them.
 The color of a quadtree node $c$ is the darkest possible shade (dark = high probability) of any point contained in the subtree rooted at $c$.
 Each node is marked with the number of points in its subtree.}
  \label{fig:quadtree-example}
 \end{figure}

Three results on quadtree properties help to establish the time complexity of quadtree operations.
They are generalized versions of our previous work~\cite[Lemmas~1 and~2]{Looz2015HRG} and state that (i) each quadtree cell contains the same expected number of points,
that (ii) the quadtree height is $O(\log n)$ whp and that (iii) the expected number of nodes in a quadtree is in $O(n)$.

\begin{lemma}
Let $\mathbb{D}_R$ be a hyperbolic or Euclidean disk of radius $R$, $j$ a probability distribution on $\mathbb{D}_R$ which fulfills the properties defined in Section~\ref{sub:notation}, $p$ a point in $\mathbb{D}_R$ which is sampled from $j$, and $T$ be a polar quadtree on $\mathbb{D}_R$.
Let $C$ be a quadtree cell at depth $i$. Then, the probability that $p$ is in $C$ is $4^{-i}$.
\label{lemma:node-cell-probabilities}
\end{lemma}
\begin{proof}
Due to the similarity of Lemma~\ref{lemma:node-cell-probabilities} to Lemma 1 of \citet{Looz2015HRG}, the proof follows a similar structure.
Let $C$ be a quadtree cell at level $k$, delimited by $\textnormal{min}_r$, $\textnormal{max}_r$, $\textnormal{min}_\phi$ and $\textnormal{max}_\phi$.
As stated in Section~\ref{sub:notation}, we require the point probability distribution to be rotationally invariant.
The probability that a point $p$ is in $C$ is then given by 
\begin{equation}
 \Pr(p \in C) = \frac{\max_{\phi} - \min_{\phi}}{2\pi} \cdot (J(\textnormal{max}_r) - J(\textnormal{min}_r))\label{eq:cell-probability-mass}.
\end{equation}
The boundaries of the children of $C$ are given by the splitting rules in Section~\ref{sub:qt-augment}.
\begin{align}
 \mathrm{mid}_\phi &:= \frac{\max_\phi+\min_\phi}{2}\label{eq:angular-split} \\
 \mathrm{mid}_r &:= (J|_{[0,R]})^{-1}\left(\frac{J(\mathrm{max}_r) + J(\mathrm{min}_r)}{2}\right)\label{eq:radial-split}
\end{align}
We proceed with induction over the depth $i$ of $C$.
Start of induction ($i$ = 0):
At depth 0, only the root cell exists and covers the whole disk.
Since $C = \mathbb{D}_R$, $\Pr(p \in C) = 1 = 4^{-0}$.

Inductive step ($i \rightarrow i+1$):
Let $C_i$ be a node at depth $i$.
$C_i$ is delimited by the radial boundaries $\mathrm{min}_r$ and $\mathrm{max}_r$, as well as the angular boundaries $\mathrm{min}_\phi$ and $\mathrm{max}_\phi$.
It has four children at depth $i+1$, separated by $\mathrm{mid}_r$ and $\mathrm{mid}_\phi$. Let $SW$ be the south west child of $C_i$.
With Eq.~(\ref{eq:cell-probability-mass}), the probability of $p\in SW$ is:
\begin{equation}
\Pr(p \in SW) = \frac{\mathrm{mid}_{\phi} - \min_{\phi}}{2\pi} \cdot \left(J\left(\mathrm{mid}_r\right) - J\left(\textnormal{min}_r\right)\right)
 \end{equation}.

Using Equations~(\ref{eq:angular-split}) and~(\ref{eq:radial-split}), this results in a probability of 
\begin{align}
\Pr(p \in SW) &= \frac{\frac{\max_\phi+\min_\phi}{2} - \min_{\phi}}{2\pi} \cdot \left(J\left((J|_{[0,R]})^{-1}\left(\frac{J(\mathrm{max}_r) + J(\mathrm{min}_r)}{2}\right)\right) - J(\textnormal{min}_r)\right)\\
\Pr(p \in SW) &= \frac{\frac{\max_\phi+\min_\phi}{2} - \min_{\phi}}{2\pi} \cdot \left(\frac{J(\mathrm{max}_r) + J(\mathrm{min}_r)}{2} - J(\textnormal{min}_r)\right)\\
\Pr(p \in SW) &= \frac{\frac{\max_\phi-\min_\phi}{2}}{2\pi} \cdot \left(\frac{J(\mathrm{max}_r) - J(\mathrm{min}_r)}{2}\right)\\
\Pr(p \in SW) &= \frac{1}{4} \frac{\max_\phi-\min_\phi}{2\pi} \cdot \left(J(\mathrm{max}_r) - J(\mathrm{min}_r)\right)
\end{align}
As per the induction hypothesis, $\Pr(p \in C_i)$ is $4^{-i}$ and $\Pr(p \in SW)$ is thus $\frac{1}{4}\cdot 4^{-i} = 4^{-(i+1)}$.
Due to symmetry when selecting $\mathrm{mid}_\phi$, the same holds for the south east child of $C_i$. Together, they contain half of the probability mass of $C_i$.
Again due to symmetry, the same proof then holds for the northern children as well.
\end{proof}
%
\begin{proposition}
 \label{THM:QUADTREE-HEIGHT}
Let $\mathbb{D}_R$ and $j$ be as in Lemma~\ref{lemma:node-cell-probabilities}.
Let $T$ be a polar quadtree on $\mathbb{D}_R$ constructed to fit $j$.
Then, for $n$ sufficiently large, $\mathrm{height}(T) \in O(\log n)$ whp.
\end{proposition}
\begin{proof}
We proved a similar lemma in previous work~\cite{Looz2015HRG}, for hyperbolic geometry only and a restricted family of probability distributions.
The requirement for that proof was that a given point $p$ has a probability of $4^{-i}$ to land in a given cell at depth $i$.
In Lemma~\ref{lemma:node-cell-probabilities}, we show that this requirement is fulfilled for the quadtrees used in this paper in both Euclidean and hyperbolic geometry.
We can thus reuse the proof of Lemma 2 of \citet{Looz2015HRG}, which we include in Appendix~\ref{sec:proof-thm-quadtree-height} since it was omitted from the conference version due to space constraints.
\end{proof}

\begin{lemma}
Just as in Lemma~\ref{lemma:node-cell-probabilities}, 
let $\mathbb{D}_R $ be a hyperbolic or Euclidean disk of radius $R$, $j$ a probability distribution on $\mathbb{D}_R$ which fulfills the properties defined in Section~\ref{sub:notation}, and $T$ be a polar quadtree on $\mathbb{D}_R$.
The expected number of  nodes in $T$ is then in $O(n)$.
\label{lemma:bound-number-quadtree-cells}
\end{lemma}

\begin{proof}
A quadtree $T$ containing $n$ points can have at most $n$ non-empty leaf cells. We can thus bound the total number of leaf cells in $T$ by limiting the number of empty cells.

An empty leaf cell occurs when a previous leaf cell $c$ is split.
We consider two cases, depending on how many of the children of $c$ contain points:

\textbf{Case 1:} All but one of the children of $c$ are empty and all points in $c$ are concentrated in one child.
We call a split of this kind an \emph{excess} split, since it did not result in dividing the points in $c$.

\textbf{Case 2:} At least two children of $c$ contain points.

The number of excess splits caused by a pair of points depends on the area they are clustered in.
Two sufficiently close points could cause a potentially unbounded number of excess splits.
However, due to Lemma~\ref{lemma:node-cell-probabilities}, each child cell contains a quarter of the probability mass of its parent cell.
Given two points $p,q$ in a cell which is split, they end up in different child cells with probability 3/4.

The expected number of excess splits for a point $p$ is thus at most\footnote{Note that the real number of excess splits might be lower, since a split might separate another point from $p$ and $q$.}
\begin{equation}
\sum_{i=0}^{\infty} i\cdot 4^{-i} = \frac{4}{9}. 
\end{equation}

Due to the linearity of expectations, the expected number of excess splits caused by $n$ points is then at most $4n/9$.
Each excess split causes four additional quadtree nodes, three of them are empty leaf cells.

If we remove all quadtree nodes caused by excess splits and reconnect the tree by connecting the remaining leaves
to their lowest unremoved ancestor, every inner node in the remaining tree $T'$ has at least two non-empty subtrees.
Since a binary tree with $n$ leaves has $O(n)$ inner nodes and the branching factor in $T'$ is at least two, $T'$ also contains at most $O(n)$ inner nodes.

Together with the expected $O(n)$ nodes caused by excess splits, this results in $O(n)$ nodes in $T$ in expectation.
\end{proof}

A direct consequence from the results above and our previous work~\cite{Looz2015HRG} is the preprocessing
time for the quadtree construction. The generalized splitting rule and storing the subtree sizes only change constant factors.

\begin{corollary}
\label{cor:qt-construction}
Since a point insertion takes $O(\log n)$ time whp, constructing a quadtree on $n$ points distributed as 
in Section~\ref{sub:notation} takes $O(n \log n)$ time whp.
\end{corollary}

\subsection{Algorithm}
\label{sub:baseline-algo}
The baseline version of our query (Algorithm~\ref{algo:quadnode-probabilistic}) has a time complexity of $\Theta(n)$, but serves as a foundation for the fast version (Section~\ref{sec:subtree-aggr}).
It takes as input a query point $q$, a function $f$ and a quadtree cell $c$.
Initially, it is called with the root node of the quadtree and recursively descends the tree.
As we prove in Prop.~\ref{lemma:independent-correct-probabilities}, the algorithm returns a point set $\neighborset(q,f) \subseteq P$ with
\begin{equation}
\label{eq:probneigh}
\Pro{p \in \neighborset(q,f)} =  f(\mathrm{dist}(q,p)).
\end{equation}

\begin{algorithm}[tb]
 \KwIn{query point $q$, prob.\ function $f$, quadtree node $c$}
 \KwOut{probabilistic neighborhood of \textit{q}}
 $\neighborset = \{\}$\;
  $\underline{b} = \dist(q, c$)\;\label{line:distanceLB} \tcc{Euclidean or hyperbolic distance between point and cell, calculated by Algorithms~\ref{algo:hyperbolic-distances} and \ref{algo:Euclidean-polar-distances}.}
 $\overline{b}$=$f(\underline{b}$)\;\label{line:probUB} \tcc{Since $f$ is monotonically decreasing, a lower bound for the distance gives an upper bound $\overline{b}$ for the probability.}
 $s$ = number of points in $c$\;
 \If{$c$ is not leaf}{\label{line:notleafcell}
 \tcc{internal node: descend, add recursive result to local set}
 \For{child $\in$ children($c$)}{
   add getProbabilisticNeighborhood($q$, $f$, child) to \neighborset \;\label{line:calling-child}
 }
}
 \Else{\tcc{leaf case: sample gaps from geometric distribution} 
 \For{i=0; $i < s$ ; i++}{\label{line:loop-iteration} \label{line:candidate-loop}
 $\delta = \ln(1-\mathit{rand}) / \ln(1-\overline{b})$\;\label{line:deltaskip}
 $i$ += $\delta$\;\label{line:jumptarget}
 \If{$i \geq s$}{
 break\;
 }
 $\mathit{prob}$ = $f(\dist(q$, c.points[$i$]))/$\overline{b}$\;\label{line:candidate-confirmation}
 add c.points[$i$] to $\neighborset$ with probability $\mathit{prob}$\label{line:add-confirmed}
 }
 }
 \Return{$\neighborset$}
 \caption{QuadNode.getProbabilisticNeighborhood}
 \label{algo:quadnode-probabilistic}
\end{algorithm}

In the following discussion, line numbers refer to lines of pseudocode in Algorithm~\ref{algo:quadnode-probabilistic}.
This query algorithm descends the quadtree recursively until it reaches the leaves.
Once a leaf $l$ is reached, a lower bound $\underline{b}$ for the distance between the query point $q$ and all the points in $l$ is computed (Line~\ref{line:distanceLB}). Such distance calculations are detailed in Appendix~\ref{sec:distance}.
Since $f$ is monotonically decreasing, this lower bound for the distance gives an upper bound $\overline{b}$ for the probability that a given point in $l$ is a member of the returned point set (Line~\ref{line:probUB}).
This bound is used to select \emph{neighbor candidates} in a similar manner as ~\citet{batagelj2005efficient}: 
In Line~\ref{line:deltaskip}, a random number of vertices is skipped, so that every vertex in $l$ is selected as a neighbor candidate with probability $\mathrm{\overline{b}}$.
The actual distance $\mathrm{dist}(q,a)$ between a candidate $a$ and the query point $q$ is at least $\mathrm{\underline{b}}$ and the probability of $a \in \neighborset(q, f)$ thus at most $\mathrm{\overline{b}}$.
For each candidate, this actual distance $\mathrm{dist}(q,a)$ is then calculated and a neighbor candidate is confirmed as a neighbor with probability $f(\mathrm{dist}(q,a))/\mathrm{\overline{b}}$ in Line~\ref{line:candidate-confirmation}.

Regarding correctness of Algorithm~\ref{algo:quadnode-probabilistic}, we can state:

\begin{proposition}
Let $T$ be a quadtree as defined above, $q$ be a query point and $f : \mathbb{R}^+ \rightarrow [0,1]$ a monotonically decreasing function
which maps distances to probabilities.
The probability that a point $p$ is returned by a PNQ ($q, f$) from Algorithm~\ref{algo:quadnode-probabilistic} is $f(\text{dist}(q,p))$,
independently from whether other points are returned.
\label{lemma:independent-correct-probabilities}
\end{proposition}
\begin{proof}
Algorithm~\ref{algo:quadnode-probabilistic} traverses the whole quadtree with all of its leaves.
Since each leaf is examined, we can concentrate on whether the points are sampled correctly within a leaf cell.
Our proof thus consists of three steps:
1) The probability that the first point in a leaf is a candidate is $\overline{b}$.
2) Given two points $p_i$ and $p_j$ in the same leaf, the probability that $p_i$ is a candidate is independent of whether $p_j$ is a candidate.
3) The probability that a point $p_i$ is a neighbor of the query point $q$ is given by Eq.~(\ref{eq:probneigh}).

Note that the hyperbolic [Euclidean] distances, which are mapped to probabilities according to the function $f$,
are calculated by Algorithm~\ref{algo:hyperbolic-distances}
[Algorithm~\ref{algo:Euclidean-polar-distances}], which are presented in Appendix~\ref{sec:distance} (together with their correctness proofs).
We continue the current proof with details for all three main steps.
\paragraph{Step 1}
Between two points, the jumping width $\delta$ is given by Line~\ref{line:deltaskip} of Algorithm~\ref{algo:quadnode-probabilistic}.
The probability that exactly $i$ points are skipped between two given candidates is $(1-\overline{b})^i\cdot \overline{b}$:
\begin{align}
 \Pr(i \leq \delta < i+1)=\Pr(i \leq \ln(1-r)/\ln(1-\overline{b})<i+1)&=\label{eq:logratio}\\
 \Pr(\ln(1-r)\leq i\cdot\ln(1-\overline{b}) \wedge \ln(1-r)>(i+1)\cdot\ln(1-\overline{b}))&=\label{eq:loginequality}\\
 \Pr(1-(1-\overline{b})^i \leq r < (1-(1-\overline{b})^{i+1}))= 1-(1-\overline{b})^{i+1} - 1+(1-\overline{b})^i&=\label{eq:uniform-r-needed}\\
 (1-\overline{b})^i(1-(1-\overline{b})) = (1-\overline{b})^i\cdot \overline{b}\label{eq:waiting-times}
\end{align}
Note that in Eq.~(\ref{eq:logratio}) the denominator is negative, thus the direction of the inequality is reversed in the transformation.
The transformation in Eq.~(\ref{eq:uniform-r-needed}) works since $r$ is uniformly distributed.

Following from Eq.~(\ref{eq:waiting-times}), the probability is $\overline{b}$ for $i=0$, and if a point is selected as a candidate, the subsequent point is selected with a probability of $\overline{b}$.

\paragraph{Step 2}
Let $p_i$, $p_j$ and $p_l$ be points in a leaf, with $i<j<l$ and let $p_i$ be a neighbor candidate.
For now we assume that no other points in the same leaf are candidates and consider the probability that $p_l$ is selected as a candidate depending on whether the intermediate point $p_j$ is a candidate.

\textbf{Case 2.1:} If point $p_j$ is a candidate, then point $p_l$ is selected if $l-j$ points are skipped after selecting $p_j$.
Due to Step 1, this probability is $(1-\overline{b})^{l-j}\cdot \overline{b}$

\textbf{Case 2.2:} If point $p_j$ is \emph{not} a candidate, then point $p_l$ is selected if $l-i$ points are skipped after selecting $p_i$.
Given that $p_j$ is not selected, at least $j-i$ points are skipped.
The conditional probability is then:
\begin{align}
 \Pr(l-i \leq \delta < l-i+1 | \delta > j-i) &=\\
 \Pr(1-(1-\overline{b})^{l-i} < r < (1-(1-\overline{b})^{l-i+1}) | \delta > j-i)&=\\
 (1-\overline{b})^{l-i}\cdot \overline{b} / (1-\overline{b})^{j-i}=
 (1-\overline{b})^{l-j}\cdot \overline{b}
\end{align}
As both cases yield the same result, the probability $\Pr(p_l \in \candidateset)$ is independent of whether $p_j$ is a candidate.

\paragraph{Step 3}
Let $C$ be a leaf cell in which all points up to point $p_i$ are selected as candidates.
Due to Step 1, the probability that $p_{i+1} $ is also a candidate, meaning no points are skipped, is $(1-\overline{b})^0\cdot \overline{b} = \overline{b}$.
Due to Step 2, the probability of $p_{i+1}$ being a candidate is independent of whether $p_i$ is a candidate.
This can be applied iteratively until the beginning of the leaf cell, yielding a probability of $\overline{b}$ for $p_i$ being a candidate, independent of whether other points are selected.
\old{Is this sufficient, or does it need to be phrased more formally, for example with an Induction?}

A neighbor candidate $p_i$ is accepted as a neighbor with probability $f(\dist(p_i, q))/\overline{b}$ in Line~\ref{line:candidate-confirmation}.
Since $\overline{b}$ is an upper bound for the neighborhood probability, the acceptance ratio is between 0 and 1.
The probability for a point $p$ to be in the probabilistic neighborhood computed by Algorithm~\ref{algo:quadnode-probabilistic} is thus:
\begin{align}
 \Pr(p\in \neighborset(q,f)) =
 \Pr(p\in \neighborset(q,f) \wedge p\in \candidateset(q,f)) &=\\
 \Pr(p \in \neighborset(q,f) | p \in \candidateset(q,f)) \cdot \Pr(p\in \candidateset(q,f))&=\\
 f(\dist(p, q))/\overline{b} \cdot \overline{b}=
 f(\dist(p, q))
\end{align}
\end{proof}

Since Algorithm~\ref{algo:quadnode-probabilistic} examines the complete quadtree, its time complexity is at least linear.
We omit a more thorough analysis until the next section, in which we show how to accelerate the query process.

\section{Queries in Sublinear Time by Subtree Aggregation}
\label{sec:subtree-aggr}
One reason for the linear time complexity of the baseline query is the fact that every quadtree node is visited.
To reach a sublinear time complexity, we thus aggregate subtrees into \emph{virtual leaf cells} whenever doing so reduces the number of examined cells and does not increase the number of candidates too much.

To this end, let $S$ be a subtree starting at depth $l$ of a quadtree $T$.
During the execution of Algorithm~\ref{algo:quadnode-probabilistic}, a lower bound $\underline{b}$ for the distance between $S$ and the query point $q$ is calculated,
yielding also an upper bound $\overline{b}$ for the neighbor probability of each point in $S$.
At this step, it is possible to treat $S$ as a \emph{virtual leaf cell}, sample jumping widths using $\overline{b}$ as upper bound and use these widths to select candidates within $S$.
Algorithm~\ref{algo:maybe-get-kth-element} is used in a virtual leaf cell where the candidate confirmation (Line~\ref{line:candidate-confirmation} of Algorithm~\ref{algo:quadnode-probabilistic}) happens in an original leaf cell.
Aggregating a subtree to a virtual leaf cell allows skipping leaf cells which do not contain candidates, but uses a weaker bound $\overline{b}$ and thus a potentially larger candidate set.
Thus, a fast algorithm requires an aggregation criterion which keeps both the number of candidates and the number of examined quadtree cells low.
\footnote{In the extreme case, candidates are selected directly at the root.
In this case, the distance to the query point is 0 and the probability bound $\overline{b}$ is $f(0)$, resulting in linearly many candidates.}
As stated before, we record the number of points in each subtree during quadtree construction.
This information is now used for the query algorithm:
We aggregate a subtree $S$ to a virtual leaf cell exactly if $|S|$, the number of points contained in $S$, is below $1 / f(\text{dist}(S,q))$.
This corresponds to less than one expected candidate within $S$.
The changes required in Algorithm~\ref{algo:quadnode-probabilistic} to use the subtree aggregation are minor.
Lines~\ref{line:notleafcell}, \ref{line:candidate-confirmation} and \ref{line:add-confirmed} are changed to:

\LinesNotNumbered
\begin{algorithm}[H]
\nlset{5} \textbf{if}\emph{ $c$ is inner node and $|c|\cdot\overline{b}\geq 1$}\textbf{ then}
\end{algorithm}

\begin{algorithm}[H]
 \nlset{14} neighbor = maybeGetKthElement($q$, $f$, $i$, $\overline{b}$, $c$)\;
 \nlset{15} add neighbor to $\neighborset$ if not empty set
\end{algorithm}
\LinesNumbered

The main change consists in the use of the function maybeGetKthElement (Algorithm~\ref{algo:maybe-get-kth-element}):

\begin{algorithm}[H]
\KwIn{query point $q$, function $f$, index $k$, bound $\overline{b}$, subtree $S$}
\KwOut{$k$th element of $S$ or empty set}
 \If{$k \geq |S|$}{
 \Return $\emptyset$\;
 }
\If{$S$.isLeaf()}{
acceptance = $f(\text{dist}(q,\mathrm{S.positions}[k]))/\overline{b} $\;
\If{$\mathit{1-rand()} <$  acceptance}{
\Return $\mathrm{S.elements}[k]$\;
}
\Else{
\Return $\emptyset$\;
}
}
\Else{\tcc{Recursive call}
offset := 0\;
\For{child $\in$ $S$.children}{
\If{$k - \mathrm{offset} < |\mathrm{child}|$}{\tcc{|child| is the number of points in \emph{child}}
\Return maybeGetKthElement($q$, $f$, $k$ - offset, $\overline{b}$, child)\;
}
offset += $|\mathrm{child}|$\;
}
}
 \caption{maybeGetKthElement}
 \label{algo:maybe-get-kth-element}
\end{algorithm}
Given a subtree $S$, an index $k$, $q$, $f$, and $\overline{b}$, the algorithm descends $S$ to the leaf cell containing the
$k$th element. This element $p_k$ is then accepted with probability $f(\dist(q, p_k)) / \overline{b}$.

Since the upper bound calculated at the root of the aggregated subtree is not smaller than the individual upper bounds at the original leaf cells, Proposition~\ref{lemma:independent-correct-probabilities} also holds for the virtual leaf cells. This establishes the correctness.

\subsection{Query Time Complexity}
\label{subsec:subtree-aggregation-complexity}
Our main analytical result of this section concerns the time complexity of the faster query algorithm.
Its proof relies on several lemmas presented afterwards.

\begin{theorem}
Let $T$ be a quadtree with $n$ points and $(q,f)$ a query pair.
A query $(q,f)$ using subtree aggregation has time complexity $O((|\neighborset(q,f)| + \sqrt{n}) \log n)$ whp.
\label{lemma:subtree-aggregation-complexity}
\end{theorem}
\begin{proof}
Similar to the baseline algorithm, the complexity of the faster query is determined by the number of recursive calls and the total number of loop iterations across the calls.
The first corresponds to the number of examined quadtree cells, the second to the total number of candidates.
With subtree aggregation, we obtain improved bounds: Lemma~\ref{lemma:subtree-aggregation-candidates} limits the number of candidates to $O(|\neighborset(q,f)| + \sqrt{n})$ whp, while Lemma~\ref{lemma:subtree-aggregation-cells} bounds the number of examined quadtree cells to $O((|\neighborset(q,f)| + \sqrt{n})\log n)$ whp.
Together, this results in a query complexity of $O((|\neighborset(q,f)| + \sqrt{n}) \log n)$ whp.
\end{proof}

For the lemmas required in the proof of Theorem~\ref{lemma:subtree-aggregation-complexity} we need to introduce some
notation:
Let $T$ be a quadtree with $n$ points, $S$ a subtree of $T$ containing $s$ points, $q$ a query point and $f$ a function mapping distances to probabilities.
The set of neighbors ($\neighborset(q,f)$), candidates ($\candidateset(q,f)$) and examined cells ($\cellset(q, f)$) are defined as in Section~\ref{sub:notation}.

For the analysis we divide the space around the query point $q$ into infinitely many bands, based on the probabilities given by $f$.
A point $p\in\pointset$ is in band $i$ exactly if the probability of it being a neighbor of $q$ is between $2^{-(i+1)}$ and $2^{-i}$:
\[
 p \in \text{band }i  \Leftrightarrow 2^{-(i+1)} < f(\text{dist}(p,q)) \leq 2^{-i}
\]
Based on these bands, we divide the previous sets into infinitely many subsets:
\begin{itemize}
 \item $\pointset(q, f, i) := \{v \in \pointset | 2^{-(i+1)} < f(\text{dist}(v,q)) \leq 2^{-i}\}$
 \item $\neighborset(q, f, i) := \neighborset(q, f) \cap \pointset(q,f,i)$
 \item $\candidateset(q, f, i) := \candidateset(q, f) \cap \pointset(q,f,i)$
 \item $\cellset(q, f, i) := \{c \in \cellset(q, f) | 2^{-(i+1)} < f(\text{dist}(c,q)) \leq 2^{-i}\}$
\end{itemize}

Note that for fixed $n$, all but at most finitely many of these sets are empty.
We call the quadtree cells in $\cellset(q,f,i)$ to be \emph{anchored} in band $i$.
The region covered by a quadtree cell is in general not aligned with the probability bands, thus a quadtree cell anchored in band $i$ ($c \in \cellset(q,f,i)$) may contain points from higher bands (i.e. with lower probabilities).

We continue with two auxiliary results used in Lemma~\ref{lemma:subtree-aggregation-candidates}.

\newcommand{\uaqcset}{\ensuremath{\Upsilon}}
\newcommand{\subtreeset}{\ensuremath{\mathcal{\Psi}}\xspace}
\begin{lemma}
Let $T$ be a polar hyperbolic [Euclidean] quadtree with $n$ points and $s < n$ a natural number.
Let $\Lambda$ be a circle in the hyperbolic [Euclidean] plane and let \subtreeset be the disjoint set of subtrees of $T$ that contain at most $s$ points and are cut by $\Lambda$.
Then, the subtrees in \subtreeset contain at most $24\sqrt{n\cdot s}$ points with probability at least $1-0.7^{\sqrt{n}}$ for $n$ sufficiently large.
\label{lemma:ring-root}
\end{lemma}
\begin{proof}
This proof is adapted from Lemma~3 of \citet{Looz2015HRG}.
Let $k := \lfloor \log_4 n/s\rfloor$ be the minimal depth at which cells have at least $s$ points in expectation.
At most $4^k$ cells exist at depth $k$, defined by at most $2^{k}$ angular and $2^{k}$ radial divisions.
When following the circumference of the query circle $\Lambda$, each newly cut cell requires the crossing of an angular or radial division.
Each radial and angular coordinate occurs at most twice on the circle boundary, thus each division can be crossed at most twice.
With two types of divisions, $\Lambda$ crosses at most $2\cdot2\cdot 2^k = 4\cdot2^{\lfloor\log_4 n/s \rfloor}$ cells at depth $k$.
Since the value of $4\cdot2^{\lfloor\log_4 n/s \rfloor}$ is at most $4\cdot2^{\log_4 n/s}$, this yields $\leq 8\cdot \sqrt{n/s}$ cut cells.
We denote the set of cut cells with \ringset.
Since the cells in \ringset cover the circumference of the circle $\Lambda$, a subtree $S$ which is cut by $\Lambda$ is either contained within one of the cells in \ringset,
corresponds to one of the cells or contains one.
In the first two cases, all points in $S$ are within the cells of \ringset.
In the second case, at least one cell of \ringset is contained in $S$.
As the subtrees are disjoint, this cell cannot be contained in any other of the considered subtrees.
Thus, there are no more subtrees containing points not in \ringset than there are cells in \ringset, which are less than $8\cdot \sqrt{n/s}$ many.

Due to Lemma~\ref{lemma:node-cell-probabilities}, the probability that a given point is in a given cell at level $k$ is $4^{-k}$.
The number of points contained in cells of \ringset thus follows a binomial distribution $B(n,p)$.
An upper bound for the probability $p$ is given by $\frac{8\cdot \sqrt{ns}}{n}$, thus a tail bound for a slightly different distribution $B(n,\frac{8\cdot \sqrt{ns}}{n})$
also holds for $B(n,p)$.
In the proof of Lemma~7 of \citet{Looz2015HRG} a similar distribution is considered.
Setting the variable $c$ to $8\sqrt{s}$, we see that the probability of \ringset containing more than $16\cdot \sqrt{sn}$ points is smaller than $0.7^{\sqrt{n}}$.

The subtrees in \subtreeset contain at most $s$ points by definition,
thus an upper bound for the number of points in these subtrees is given by $s\cdot 8\cdot \sqrt{n/s}$ (points not in \ringset) + $16\cdot \sqrt{sn}$ (points in \ringset).
This results in at most $24\cdot \sqrt{sn}$ points contained in \subtreeset with probability at least $1-0.7^{\sqrt{n}}$.
\end{proof}

\begin{lemma}
Let $n$ be a natural number and let $A$, $B$ be sets with $A \subseteq B, |B| \leq n$ and the following property: $\Pr(b \in A) \geq 0.5$, $\forall b \in B$.
Further, let the probabilities for membership in $A$ be independent.
Then, the number of points in $B$ is in $O(|A| + \log n)$ with probability at least $1-1/n^3$.
\label{lemma:half-prob-set}
\end{lemma}
\begin{proof}
Let $X = |A|$ be a random variable denoting the size of $A$.
Since the individual probabilities for membership in $A$ might be different, $X$ does not necessarily follow a binomial distribution.
We define an auxiliary distribution $Y := B(|B|, 0.5)$.
Since all membership probabilities for $A$ are at least 0.5, lower tail bounds derived for $Y$ also hold for $X$.

The probability that $Y$ is less than $0.1|B|$ is then~\cite{Hoe63}:
\begin{align}
\Pr(Y < 0.1|B|) &\leq \exp\left(-2\frac{(0.5|B|-0.1|B|)^2}{|B|} \right) = \exp\left(-0.32|B|\right)\\ 
\end{align}

If $|B| \leq 10\log n$, then $|B|$ is trivially in $O(\log n)$, otherwise the probability $\Pr(|A| < 0.1|B|)$ is $\Pr(|A| < 0.1|B|) \leq \Pr(Y < 0.1|B|) \leq \exp\left(-3.2\log n\right) = n^{-3.2} < 1/n^3$.
Thus $|B| \leq 10|A| \in O(|A|)$ with probability at least $1-1/n^3$.
\end{proof}

The following Lemmas~\ref{lemma:subtree-aggregation-candidates} and~\ref{lemma:subtree-aggregation-cells} bound the number of examined candidates and examined quadtree cells, concluding this proof of Theorem~\ref{lemma:subtree-aggregation-complexity}.

\begin{lemma}
Let $T$ be a quadtree with $n$ points and $(q,f)$ a query pair.
The number of candidates examined by a query using subtree aggregation is in $O(|\neighborset(q,f)| + \sqrt{n})$ whp.
\label{lemma:subtree-aggregation-candidates}
\end{lemma}
\begin{proof}
For the analysis we consider each probability band $i$ separately.
As defined above, band $i$ contains points with a neighbor probability of $2^{-(i+1)}$ to $2^{-i}$.
Among the cells anchored in band $i$, some are original leaf cells and others are virtual leaf cells created by subtree aggregation.
The virtual leaf cells contain less than one expected candidate and thus less than $2^{i+1}$ points. The capacity of the original leaf cells is constant.
All the points in cells anchored in band $i$ have a probability between $2^{-(i+1)}$ and $2^{-i}$ to be a candidate.
Among the points in virtual or original leaf cells, some are in the same band their cell is anchored in, others are in higher cells.

We divide the set of points within cells anchored in band $i$ into four subsets:
\begin{enumerate}
 \item points in band $i$ and in original leaf cells
 \item points in band $i$ and in virtual leaf cells
 \item points not in band $i$ and in original leaf cells
 \item points not in band $i$ and in virtual leaf cells
\end{enumerate}

The points in the first two sets are unproblematic.
Since the probability that a point in these sets is a neighbor is at least $2^{-(i+1)}$, the probability for a given candidate to be a neighbor is at least $\frac{1}{2}$.
Due to Lemma~\ref{lemma:half-prob-set}, the number of candidates in these sets is in $O(|\neighborset(q,f)| + \log n)$ whp, which is in $O(|\neighborset(q,f)| + \sqrt{n})$ whp.

Points in the third set are in cells cut by the boundary between band $i$ and band $i+1$.
Since the probabilities are determined by the distance, this boundary is a circle and we can use Lemma~\ref{lemma:ring-root} to bound the number of points to $24\sqrt{n\cdot \mathrm{capacity}}$ with probability at least $1-0.7^{\sqrt{n}}$ for $n$ sufficiently large.
The mentioned capacity is the capacity of the original leaf cells.

Likewise, points in the fourth set are in virtual leaf cells cut by the boundary between bands $i$ and $i+1$.
A virtual leaf cell, which is an aggregated subtree, contains at most $2^{i+1}$ points, otherwise it would not have been aggregated.
Again, using Lemma~\ref{lemma:ring-root}, we can bound the number of points in these sets to $24\sqrt{n\cdot 2^{i+1}}$ points with probability at least $1-0.7^{\sqrt{n}}$.

We denote the union of the third and fourth sets with $\overhangset(q,f,i)$.
From the individual bounds derived in the previous paragraphs, we obtain an upper bound for the number of points in $\overhangset(q,f,i)$ of $24(\sqrt{n\cdot \mathrm{capacity}} + \sqrt{n\cdot 2^{i+1}})$ with probability at least $(1-0.7^{\sqrt{n}})^2$.
Simplifying the bound, we get that $|\overhangset(q,f,i)| \leq 24 \sqrt{n}\cdot(2^{(i+1)/2} + \sqrt{\mathrm{capacity}})$ with probability at least $1-2\cdot0.7^{\sqrt{n}}$.

Each of the points in $\overhangset(q,f,i)$ is a candidate with a probability between $2^{-i}$ and $2^{-(i+1)}$.
The candidates are sampled independently (see Step 2 of Lemma~\ref{lemma:independent-correct-probabilities}). 
While different points may have different probabilities of being a candidate and the total number of candidates does not follow a binomial distribution,
we can bound the probabilities from above with $2^{-i}$.

We proceed towards a Chernoff bound for the total number of candidates across all overhangs.
Let $X_i$ denote the random variable representing the candidates within $|\overhangset(q,f,i)|$ and let $X = \sum_{i=0}^{\infty} X_i$ denote the total number of candidates in overhangs.

The expected value $\mathbb{E}(X)$ follows from the linearity of expectations:
\begin{align}
\mathbb{E}(X) &= \sum_{i=0}^{\infty} \mathbb{E}(X_i)\\
&= \sum_{i=0}^{\infty} 24 \sqrt{n}\cdot(2^{(i+1)/2} + \sqrt{\mathrm{capacity}})\cdot2^{-i})\\
&= 24 \sqrt{n} \sum_{i=0}^{\infty} \sqrt{2}\cdot 2^{-i/2} + 2^{-i}\sqrt{\mathrm{capacity}}))\\
&= 24 \sqrt{n} ((2\sqrt{2}+2) + 2\sqrt{\mathrm{capacity}})
\end{align}

(Cells anchored in the band $\infty$, which has an upper bound $\overline{b}$ of zero for the neighborhood probability, do not have any candidates and can be omitted here.)

Since the candidates are sampled independently with a probability of at most $2^{-i}$, we can treat $X$ as a sum of independent Bernoulli random variables without loosing generality.
This allows us to use a multiplicative Chernoff bound~\cite{mitzenmacher2005probability} and we can now give an upper bound for the probability that the overhangs contain more than twice as many candidates as expected:
\begin{align}
\Pr(X > 2\mathbb{E}(X)) &\leq \left( \frac{e}{2^2} \right)^{\mathbb{E}(X)}\\
&= \left( \frac{e}{2^2} \right)^{24 \sqrt{n} ((2\sqrt{2}+2) + 2\sqrt{\mathrm{capacity}})}\\
&\leq \left( \frac{e}{2^2} \right)^{\sqrt{n}}\\
&\leq 0.7^{\sqrt{n}}
\end{align}

Including this last one, we have a chain of $2n+1$ tail bounds, each with a probability of at least $(1-0.7^{\sqrt{n}})$.
The event that any of these tail bounds is violated is a union over each event that a specific tail bound is violated.
With a union bound~\cite[Lemma 1.2]{mitzenmacher2005probability}, the probability that any of the individual tail bounds is violated is at most $(2n+1)0.7^{\sqrt{n}}$.
Since $\frac{1}{(2n+1)0.7^{\sqrt{n}}}$ grows faster than $n$ for $n$ sufficiently large,
we conclude that the total number of candidates is thus bounded by $O(|\neighborset(q,f)|) + 48\sqrt{n}((2\sqrt{2}+2) + 2\sqrt{\mathrm{capacity}})$ with probability at least $(1-1/n)$ for $n$ sufficiently large.
The leaf capacity is constant, thus the number of candidates evaluated during execution of a query $(q,f)$ is in $O(|\neighborset(q,f,i)| + \sqrt{n})$ whp.
\end{proof}

We proceed with a lemma necessary for bounding the number of examined quadtree cells in a query.

\begin{lemma}
Let $T$ be a quadtree with $n$ points and $(q,f)$ a query pair.
The number of quadtree cells examined by a query using subtree aggregation is in $O((|\neighborset(q,f)| + \sqrt{n}) \log n)$.
\label{lemma:subtree-aggregation-cells}
\end{lemma}

To prove Lemma~\ref{lemma:subtree-aggregation-cells}, we first need to introduce another auxiliary lemma:
\begin{lemma}
Let $\mathbb{D}_R$ be a hyperbolic or Euclidean disk of radius $R$ and let $T$ be a polar quadtree on $\mathbb{D}_R$ 
containing $n$ points distributed according to Section~\ref{sub:notation}.
Let \uaqcset(q,f) be the set of unaggregated quadtree cells that have only (virtual) leaf cells as children.
 With a query using subtree aggregation, $|\uaqcset(q,f)|$ is in $O(|\neighborset(q,f)| + \sqrt{n})$ whp.
 \label{lemma:bound-unaggregated-quadtree-cells-with-leaf-children}
\end{lemma}
\begin{proof}
 Let $c\in \uaqcset(q,f,i)$ be such an unaggregated quadtree cell anchored in band $i$ that has only original or virtual leaf cells as children.
 It contains at least $2^i$ points and has four children, of which at least one is also anchored in band $i$.
 We denote this (virtual) leaf anchored in band $i$ with $l$.
 Since each child of $c$ contains the same probability mass (Lemma~\ref{lemma:node-cell-probabilities}), each point of $c$ is in $l$ with probability $1/4$:
\begin{equation}
 \Pr(p \in l | p \in c) = \frac{1}{4}.
\end{equation}

A point in $l$ is a candidate (in $l$) with probability $f(\dist(q,l))$, which is between $2^{-(i+1)}$ and $2^{-i}$ since $l$ is anchored in band $i$.
The probability that a given point $p\in c$ is a candidate in $l$ is then 
\begin{equation}
 \Pr(p \in l \wedge p \in \candidateset(q,f,i) | p \in c) = \frac{1}{4}\cdot f(\dist(q,l)) \geq 2^{-(i+3)}
\end{equation}

Since the point positions and memberships in $\candidateset(q,f,i)$ are independent, we can bound the number of candidates in $l$ with a binomial distribution $B(|c|, 2^{-(i+3)})$.
The probability that $l$ contains no candidates is:
\begin{align}
 f\left(0, |c|, \frac{1}{8} \cdot 2^{-i}\right) &= \left(1-\frac{1}{8} \cdot 2^{-i}\right)^{|c|} \leq \left(1-\frac{1}{8} \cdot \frac{1}{2^i}\right)^{2^i} 
\end{align}

Considered as a function of $i$, this probability is monotonically ascending.
In the limit of $2^i \rightarrow \infty$, it trends to $\exp(-1/8) \approx 0.88$, a value it never exceeds.
The probability that the cell $c$ contains at least one candidate is then above $1-\frac{1}{\sqrt[8]{e}} > 0.1$.

For each cell in \uaqcset, the probability that it contains at least one candidate is $> 0.1$.
Let $X$ be the random variable denoting the number of cells in \uaqcset that contain at least one candidate.
We define an auxiliary binomial distribution $B(|\uaqcset|, 0.1)$ and use a tail bound to estimate the number of cells in \uaqcset containing candidates.
Let $Y\propto B(|\uaqcset|, 0.1)$ be a random variable distributed according to this auxiliary distribution.

We use a tail bound from \citet{ArratiaGordon1989} to limit the probability that $Y < 0.05|\uaqcset|$ to at most $\exp(-|\uaqcset|/80)$.
Since $0.1$ was a lower bound for the probability that a cell contains a candidate, this tail bound also holds for $X$.
The probability that the set of \uaqcset{} contains at least $0.05|\uaqcset|$ many candidates is then at least $(1-\exp(-|\uaqcset|/80))$.

If $|\uaqcset| \in o(\sqrt{n})$, then $|\uaqcset|$ is trivially in $O(\sqrt{n})$, if $|\uaqcset| \in \omega(\sqrt{n})$,
then the probability $(1-\exp(-|\uaqcset|/80))$ is smaller than $(1-\exp(-\sqrt{n}/80))$, which is  $< 1/n$ for sufficiently large $n$.
Thus the number of examined quadtree cells during a query is then linear in the number of candidates.
Due to Lemma~\ref{lemma:subtree-aggregation-candidates}, this is in $O(|\neighborset(q,f)| + \sqrt{n})$.
\end{proof}

The proof of Lemma~\ref{lemma:subtree-aggregation-cells} then follows easily:
\begin{proof}
We split the set of examined quadtree cells into three categories:
\begin{itemize}
 \item leaf cells and root nodes of aggregated subtrees $(C1)$
 \item parents of cells in the first category $(C2)$
 \item all other $(C3)$
\end{itemize}
The third category $(C3)$ then exclusively consists of inner nodes in the quadtree. When following a chain of nodes in category $C3$ from the root downwards, it ends with a node in category $C2$.
The size $|C3|$ is thus at most $O(|C2| \log n)$ whp, since the number of elements in a chain cannot exceed the height of the quadtree, which is $O(\log n)$ by Proposition~\ref{THM:QUADTREE-HEIGHT}.

With a branching factor of 4, $|C1| = 4|C2|$ holds.

The number of cells in category $C2$ can be bounded using Lemma~\ref{lemma:bound-unaggregated-quadtree-cells-with-leaf-children} to $O(|\neighborset(q,f)| + \sqrt{n})$ with high probability.
The total number of examined cells is thus in $O((|\neighborset(q,f)| + \sqrt{n}) \log n)$.
\end{proof}

\section{Experimental Evaluation}
To evaluate the empirical performance of our probabilistic neighborhood query algorithm, we use it to process dynamic updates to random hyperbolic graphs.
We compare the running times of our quadtree structure with those of the cell data structure of \citet{bringmann2015geometric}, which gives the best performance in static generation.

\paragraph{Implementations}
Our implementation of the quadtree data structure uses the NetworKit toolkit~\cite{networkit2016journal}. 
For the cell data structure of Bringmann et al., we use the implementation of Anton Krohmer\footnote{\url{https://bitbucket.org/HaiZhung/hyperbolic-embedder}, changes based on commit f83b46111c69819b7447fbd29fe8ed9bdb1fba3f}, developed for a hyperbolic embedding publication~\cite{blsius_et_al:LIPIcs:2016:6367} which won the Track B best paper award of ESA 2016.
We implemented node addition and removal as well as a dynamic generator on top of this cell data structure\footnote{A tree-based data structure available in the same implementation already offers node addition and removal, but no constant-time random access to individual nodes, as it uses STL sets in leaf nodes.
While it is in principle possible to replace the sets with our tree data structure allowing logarithmic random access, it would defeat the purpose of this comparison:
Replacing another data structure, as a preparation to a comparison, with the data structure it is being compared to, will not offer much insight.}.

\paragraph{Experimental Setup}
Both implementations are written in C++11 and compiled with GCC 4.8.2. Our implementation is available as part of NetworKit 4.1.
Running time measurements were made on a single Intel Xeon E5-2680 core clocked at 2.70 GHz.

In each iteration of the benchmark, one point is deleted from the data structure, moved to a random location consistent with the probability distribution and reinserted into the data structure.
The graph is then updated with the new position. For each graph, we execute 10000 point movements and recreate the graph after each.

The temperature parameter $T$ was set to 0.1, the dispersion parameter $\alpha$ to 0.75 and the radius $R$ to $2\cdot \log n -1$, leading to an average degree of $\approx 9.3$.

\paragraph{Results}
\definecolor{markedcolor}{RGB}{31,120,180}
\definecolor{plottinggreen}{RGB}{178,223,138}
\definecolor{thirdhue}{RGB}{228,26,28}

\begin{figure}
\centering[tb]
\begin{tikzpicture}
 \begin{axis}[xmode=log, ymode=log, legend pos = north west,ylabel= ms / iteration, xlabel=n]
   \addplot[markedcolor,mark=*, only marks] table[x = n, y expr = \thisrow{dyn}  / \thisrow{iter}] {plots/pnq-dynamic-log-average};\addlegendentry{PNQs};
   \addplot[markedcolor] expression[domain=8192:536870912] {(0.931*sqrt(x) + 840) / 10000};\addlegendentry{$a\cdot\sqrt{n}+b$};
   \addplot[plottinggreen,mark=diamond*, only marks] table[x = n, y expr = \thisrow{dyn} / \thisrow{iter}] {plots/kromer-dynamic-log};\addlegendentry{\citet{blsius_et_al:LIPIcs:2016:6367}}
   \addplot[plottinggreen] expression[domain=8192:536870912] {(0.00043*x + 0.02*sqrt(x) + 76) / 1000};\addlegendentry{$c\cdot n + d\cdot\sqrt{n}+e$};
 \end{axis}
\end{tikzpicture}
\caption{Running time of dynamic node movements, values are averaged over 10000 movements.
The quadtree operations are up to two orders of magnitude faster and scale better with increasing graph size.
Trend lines are fitted with $a=0.931$, $b=840$, $c=0.00043$, $d=0.02$ and $e=76$.}
\label{plot:dynamic-running-times}
\end{figure}
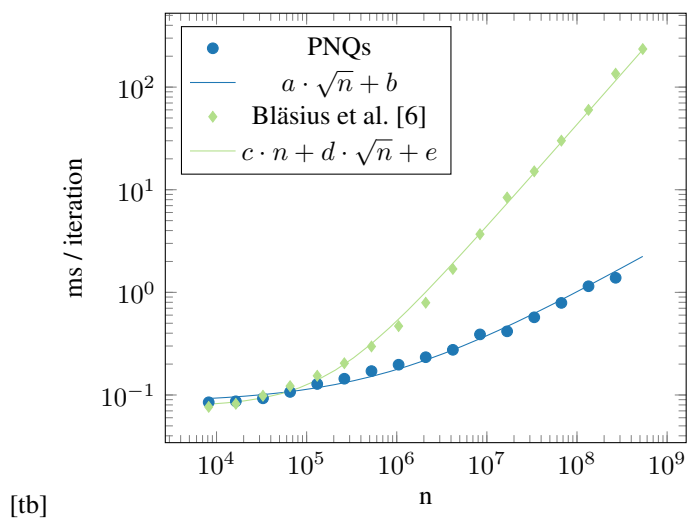

Figure~\ref{plot:dynamic-running-times} shows the experimental time measurements with sparse graphs of varying sizes.
Our method is faster for graphs of at least $10^5$ nodes; the improvement reaches two orders of magnitude
for graphs with $2\cdot 10^8$ nodes; a query on hundreds of milllions of vertices returning about 10 neighbors runs in the order of milliseconds.

\begin{figure}[tb]
 \centering
\begin{tikzpicture}
 \begin{axis}[xmode=log, ymode=log, axis equal, legend pos = north west,ylabel= node movements, xlabel=n]
   \addplot[smooth,thirdhue,mark=*] table[x = n, y expr = (\thisrow{construct} + \thisrow{trim}) / ((\thisrow{nativedyn} - \thisrow{dyn})  / \thisrow{iter})] {plots/pnq-dynamic-plus-qt-construction-2};
 \end{axis}
\end{tikzpicture}
\caption{Number of node movements needed to amortize overhead of quadtree construction. Values are averaged over 10000 iterations.}
\label{plot:qt-construction-amortization}
\end{figure}
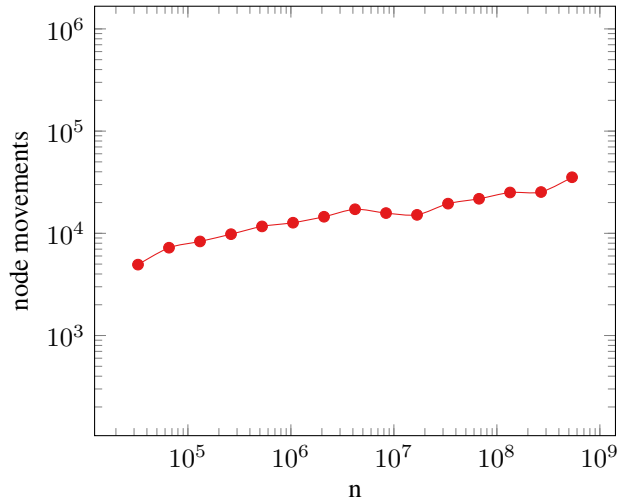

The fastest way to obtain a static graph together with a sequence of dynamic updates would be to generate the static graph first with the implementation of \citet{blsius_et_al:LIPIcs:2016:6367}, then the dynamic updates with our data structure.
Moving data into a new data structure takes at least linear time, Figure~\ref{plot:qt-construction-amortization} shows the number of node movements needed until the overhead of this preprocessing step is amortized by the faster queries.
For graphs with $10^5$ to $10^8$ vertices, the quadtree queries are faster if performing more than roughly $10^4$ iterations, a value that grows only slowly with increasing graph size.

The expected degree of a vertex in a random hyperbolic graph depends on its radials coordinate, a smaller radius leading to a higher degree.
Since the edge probabilities are symmetrical, vertices with small radius are more likely to be in the result set of a query.
Due to this effect, the central cells in the quadtree are examined much more often than cells on the periphery, see Figure~\ref{plot:cell-query-count}.
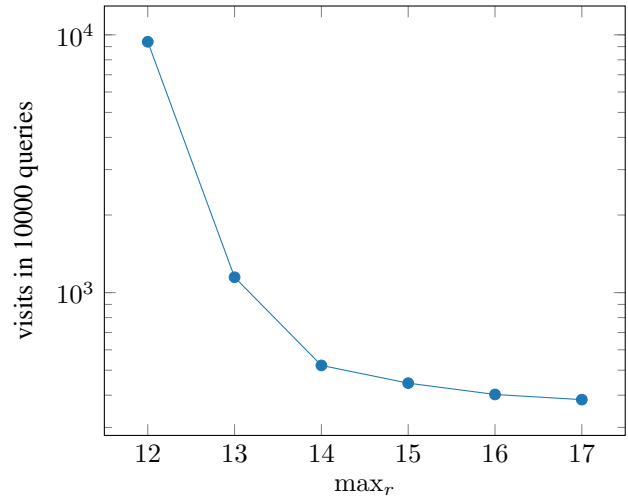
\begin{figure}
\centering
 \begin{tikzpicture}
  \begin{axis}[ymode=log,xlabel=$\max_r$,ylabel=visits in 10000 queries]
   \addplot[markedcolor,mark=*] table[ x = radius, y expr = \thisrow{queries}] {plots/quadtree-cell-query-count};
  \end{axis}
 \end{tikzpicture}
 \caption{Probability that a given cell is examined in an update step, depending on its maximum radial coordinate $\max_r$.
 Cells in the center of the polar disk are visited almost certainly, while cells in outer regions are visited rarely. Measurements are made on a random hyperbolic graph with $2^{13}$ vertices.}
 \label{plot:cell-query-count}
\end{figure}

Probably due to this effect, the query time can be significantly improved by deliberately imbalancing the polar quadtree and allocating less probability mass to the inner children.
The splitting radius $\mathrm{split}_\hyperbolic{}$ which divides the outer from the inner children of $T$, originally given in Eq.~\eqref{eq:radial-split}, is then governed by a balance parameter $b$:
\begin{equation}
\text{mid}_{r\hyperbolic{}} = \acosh((1-b)\cdot \cosh(\alpha\cdot\text{max}_{r_\hyperbolic{}}) + b\cdot \cosh(\alpha\cdot\text{min}_{r_\hyperbolic{}}))/\alpha
 \label{eq:splitting-nodes}
\end{equation}
The original behavior of Eq.~\eqref{eq:radial-split} is equivalent to setting $b$ to 0.5.
Choosing instead $b=0.001$, which yields an allocation of 0.1\% of the area to the inner two children and 99.9\% to the outer children, decreases running time by more than an order of magnitude compared to a balanced tree (Figure~\ref{plot:balance-benchmark}).
\begin{figure}
\centering
 \begin{tikzpicture}
  \begin{axis}[xmode=log, ymode=log, xlabel=share of area in inner children, ylabel = ms / query]
   \addplot[markedcolor,mark=*] table[x expr = 1 - \thisrow{balance}, y expr = \thisrow{time} / \thisrow{iter}] {plots/pnq-benchmark-balance};
  \end{axis}
 \end{tikzpicture}
 \caption{Influence of balance parameter on running time. Measurements are for a graph with $2^{23}$ vertices and averaged over 10000 queries.
 Deliberately imbalancing the quadtree improves running times by over one order of magnitude.}
 \label{plot:balance-benchmark}
\end{figure}
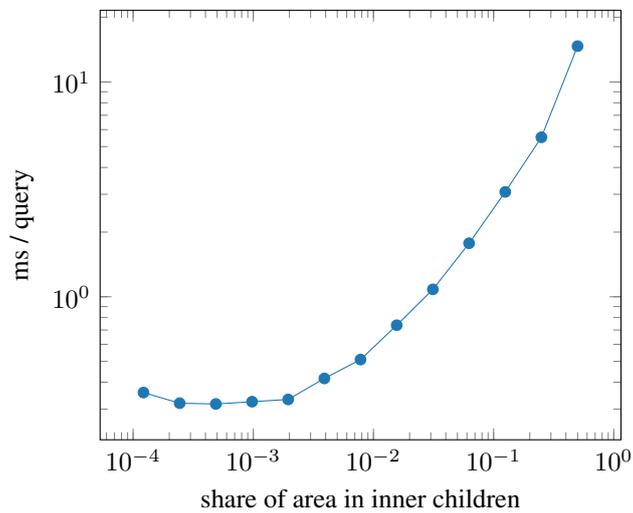

\section{Conclusions}
\label{sec:conclusion}

We described a new, gradual dynamic model for random hyperbolic graphs and proved its consistency.

To process dynamic graph updates, we formally defined the notion of probabilistic neighborhoods and presented a quadtree-based query algorithm for such neighborhoods in the Euclidean and hyperbolic plane.
Our analysis shows a time complexity of $O((|N(q,f)| + \sqrt{n})\log n)$ for $n$ points and a result set $|N(q,f)|$.
To our knowledge, our algorithm is the first to solve the problem asymptotically faster than pairwise distance probing.
These queries can be used to sample fast updates in dynamic models for RHGs.
In our experiments, our query algorithm is faster by up to two orders of magnitude than the data structure by \citet{blsius_et_al:LIPIcs:2016:6367}.
The proposed dynamic model has already been used in scaling experiments for dynamic network analysis algorithms~\cite{doi:10.1080/15427951.2016.1177802}.

The complexity results for probabilistic neighborhood queries hold for Euclidean geometry as well, making the query algorithm applicable to other sampling problems in spatial datasets and thus of independent interest.

\bibliographystyle{plainnat}
\bibliography{Bibliography}

\clearpage
\appendix

\newcommand{\lE}{\ensuremath{\mathrm{leftExtremum}}}
\newcommand{\rE}{\ensuremath{\mathrm{rightExtremum}}}

\section{Proof of Lemma 2 of von Looz et al. [2015]}
\label{sec:proof-thm-quadtree-height}
\begin{proof}
In a complete quadtree, $4^i$ cells exist at depth $i$. For analysis purposes only, we construct such 
a complete but initially empty quadtree of height $k = 3\cdot\lceil \log_4(n)\rceil$, which has at least $n^3$ leaf cells.
As seen in Lemma~\ref{lemma:node-cell-probabilities}, a given point has an equal chance to land in each leaf cell.
Hence, we can apply Lemma 6 of \citet{Looz2015HRG} with each leaf cell being a bin and a point being a ball.
(The fact that we can have more than $n^3$ leaf cells only helps in reducing the average load.)
From this we can conclude that, for $n$ sufficiently large, no leaf cell of the current tree contains more than 1 point with high probability (whp).
Consequently, the total quadtree height does not exceed $k = 3\cdot\lceil \log_4(n)\rceil \in O(\log n)$ whp.

Let $T'$ be the quadtree as constructed in the previous paragraph, starting with a complete quadtree of height~$k$ and splitting leaves when their capacity is exceeded.
Let $T$ be the quadtree created in our algorithm, starting with a root node, inserting points and also splitting leaves when necessary, growing the tree downward.

Since both trees grow downward as necessary to accommodate all points, but $T$ does not start with a complete quadtree of height~$k$, the set of quadtree nodes in $T$ is a subset of the quadtree nodes in $T'$.
Consequently, the height of $T$ is bounded by $O(\log n)$ whp as well.
\end{proof}

\section{Distance between Quadtree Cell and Point}
\label{sec:distance}
To calculate the upper bound $\overline{b}$ used in Algorithm~\ref{algo:quadnode-probabilistic}, we need a lower bound for the distance between the query point $q$ and any point in a given quadtree cell.
Since the quadtree cells are polar, the distance calculations might be unfamiliar and we show and prove them explicitly.
For the hyperbolic case, the distance calculations are shown in Algorithm~\ref{algo:hyperbolic-distances} and proven in Lemma~\ref{lemma:hyperbolic-distances}.
The Euclidean calculations are shown in Algorithm~\ref{algo:Euclidean-polar-distances} and proven in Lemma~\ref{lemma:Euclidean-polar-distances}.

\LinesNumbered
\begin{algorithm}
\KwIn{quadtree cell $C$ = ($\text{min}_r$, $\text{max}_r$, $\text{min}_\phi$, $\text{max}_\phi$), query point $q = (\phi_q, r_q)$}
\KwOut{infimum and supremum of hyperbolic distances $q$ to interior of $C$}
\tcc{start with corners of cell as possible extrema}
cornerSet = \{($\text{min}_\phi$, $\text{min}_r$), ($\text{min}_\phi$, $\text{max}_r$), ($\text{max}_\phi$, $\text{min}_r$), ($\text{max}_\phi$, $\text{max}_r$)\}\;
a = $\cosh(r_q)$\;
b = $\sinh{r_q}\cdot\cos(\phi_q-\text{min}_\phi)$\;
\tcc{Left/Right boundaries}
leftExtremum = $\frac{1}{2}\ln\left(\frac{a+b}{a-b}\right)$\label{line:left-extremum}\;
\If{$\text{min}_r < \lE < \text{max}_r$}{
add $(\text{min}_\phi, \lE)$ to cornerSet\;
}

b = $\sinh{r_q}\cdot\cos(\phi_q-\text{max}_\phi)$\;
rightExtremum = $\frac{1}{2}\ln\left(\frac{a+b}{a-b}\right)$\label{line:right-extremum}\;
\tcc{Top/bottom boundaries}
\If{$\text{min}_r < \rE < \text{max}_r$}{
add $(\text{max}_\phi, \rE)$ to cornerSet\;
}

\If{$\text{min}_\phi < \phi_q \text{max}_\phi$}{
add $(\phi_q, \text{min}_r)$ and $(\phi_q, \text{max}_r)$ to cornerSet\;
}
$\phi_{\text{mirrored}} = \phi_q + \pi \mod 2\pi$\;

\If{$\text{min}_\phi < \phi_{\text{mirrored}} < \text{max}_\phi$}{
add $(\phi_{\text{mirrored}}, \text{min}_r)$ and $(\phi_{\text{mirrored}}, \text{max}_r)$ to cornerSet\;
}
\tcc{If point is in cell, distance is zero:}
\If{$\text{min}_\phi \leq \phi_q < \text{max}_\phi \text{ AND }\text{min}_r \leq r_q < \text{max}_r$}{
infimum = 0\;
}
\Else{
infimum = $\min_{e \in \text{cornerSet}} \text{dist}_\hyperbolic{} (q, e)$\;
}
supremum = $\max_{e \in \text{cornerSet}} \text{dist}_\hyperbolic{} (q, e)$\;
\Return{infimum, supremum};

 \caption{Infimum and supremum of distance in a hyperbolic polar quadtree}
 \label{algo:hyperbolic-distances}
\end{algorithm}

\begin{algorithm}
\KwIn{quadtree cell $C$ = ($\text{min}_r$, $\text{max}_r$, $\text{min}_\phi$, $\text{max}_\phi$), query point $q = (\phi_q, r_q)$}
\KwOut{infimum and supremum of Euclidean distances $q$ to interior of $C$}
\tcc{start with corners of cell as possible extrema}
 cornerSet = \{($\text{min}_\phi$, $\text{min}_r$), ($\text{min}_\phi$, $\text{max}_r$), ($\text{max}_\phi$, $\text{min}_r$), ($\text{max}_\phi$, $\text{max}_r$)\}\;
\tcc{Left/Right boundaries}
\lE = $r_q\cdot\cos(\text{min}_\phi - \phi_q)$\;\label{line:left-extremum-Euclidean}
\If{$\text{min}_r < \lE < \text{max}_r$}{
add $(\text{min}_\phi, \lE)$ to cornerSet\;
}

\rE = $r_q\cdot\cos(\text{max}_\phi - \phi_q)$\;\label{line:right-extremum-Euclidean}
\If{$\text{min}_r < \rE < \text{max}_r$}{
add $(\text{max}_\phi, \rE)$ to cornerSet\;
}

\tcc{Top/bottom boundaries}
\If{$\text{min}_\phi < \phi_q < \text{max}_\phi$}{
add $(\phi_q, \text{min}_r)$ and $(\phi_q, \text{max}_r)$ to cornerSet\;
}
$\phi_{\text{mirrored}} = \phi_q + \pi \mod 2\pi$\;

\If{$\text{min}_\phi < \phi_{\text{mirrored}} < \text{max}_\phi$}{
add $(\phi_{\text{mirrored}}, \text{min}_r)$ and $(\phi_{\text{mirrored}}, \text{max}_r)$ to cornerSet\;
}
\tcc{If point is in cell, distance is zero:}
\If{$\text{min}_\phi \leq \phi_q < \text{max}_\phi \text{ AND }\text{min}_r \leq r_q < \text{max}_r$}{
infimum = 0\;
}
\Else{
infimum = $\min_{e \in \text{cornerSet}} \text{dist}_\hyperbolic{} (q, e)$\;
}
supremum = $\max_{e \in \text{cornerSet}} \text{dist}_\hyperbolic{} (q, e)$\;
\Return{infimum, supremum};
 \caption{Infimum and supremum of distance in a Euclidean polar quadtree}
 \label{algo:Euclidean-polar-distances}
\end{algorithm}


\begin{lemma}
 Let $C$ be a quadtree cell and $q$ a point in hyperbolic space.
 The first value returned by Algorithm~\ref{algo:hyperbolic-distances} is the distance of $C$ to $q$.
 \label{lemma:hyperbolic-distances}
\end{lemma}

\begin{proof}
When $q$ is in $C$, the distance is trivially zero.
Otherwise, the distance between $q$ and $C$ can be reduced to the distance between $q$ and the boundary of $C$, $\delta C$:
\begin{equation}
\text{dist}_\hyperbolic{} (C, q) = \text{dist}_\hyperbolic{} (\delta C, q) = \inf_{p \in \delta C} \text{dist}_\hyperbolic{} (p, q) 
\end{equation}
Since the boundary is closed, this infimum is actually a minimum:
\begin{equation}
\text{dist}_\hyperbolic{} (C, q) = \inf_{p \in \delta C} \text{dist}_\hyperbolic{} (p, q) = \min_{p \in \delta C} \text{dist}_\hyperbolic{} (p, q) 
\end{equation}
The boundary of a quadtree cell consists of four closed curves:
\begin{itemize}
 \item left: $\{(\text{min}_\phi, r) |  \text{min}_r \leq r \leq \text{max}_r \} $
 \item right: $\{(\text{max}_\phi, r) |  \text{min}_r \leq r \leq \text{max}_r \} $
 \item lower: $\{(\phi, \text{min}_r) | \text{min}_\phi \leq \phi \leq \text{max}_\phi \} $
 \item upper: $\{(\phi, \text{max}_r) | \text{min}_\phi \leq \phi \leq \text{max}_\phi \} $
\end{itemize}
We write the distance to the whole boundary as a minimum over the distances to its parts:
\begin{equation}
 \text{dist}_\hyperbolic{} (\delta C, q) = \min_{A \in \{\text{left, right, lower, upper} \}} \text{dist}_\hyperbolic{} (A, q)
\end{equation}

All points on an angular boundary curve $A$ have the same angular coordinate $\phi_A$.
Let $d_A(r) = \mathrm{acosh}(\cosh(r)\cosh(r_q) - \sinh(r)\sinh(r_q) \cos(\phi_q - \phi_A))$ for a fixed point $q$.
The distance $\text{dist}_\hyperbolic{} (A, q)$ can then be reduced to:
\begin{align}
 \text{dist}_\hyperbolic{} (A, q) &= \min_{\text{min}_r \leq r \leq \text{max}_r} d_A(r)\\
\end{align}
The minimum of $d_A$ on $A$ is the minimum of $d_A(\text{min}_r),$  $d_A(\text{max}_r)$ and the value at possible extrema.
To find the extrema, we define a function $g(r) = \cosh(r)\cosh(r_q) - \sinh(r)\sinh(r_q) \cos(\phi_q - \phi_A)$.
Since $\mathrm{acosh}$ is strictly monotone, $g(r)$ has the same extrema as $d_A(r)$.

The factors $\cosh(r_q)$ and $\sinh(r_q) \cos(\phi_q - \phi_A)$ do not depend on $r$, to increase readability we substitute them with the constants $a$ and $b$:
\begin{align}
 a &= \cosh(r_q)\\
 b &= \sinh(r_q) \cos(\phi_q - \phi_A)\\
 d_A(r) &= \mathrm{acosh}(\cosh(r)\cdot a - \sinh(r)\cdot b)\\
 g(r) &= \cosh(r)\cdot a - \sinh(r)\cdot b
\end{align}
The derivative of $g$ is thus:
\begin{equation}
g'(r) = \sinh(r)\cdot a - \cosh(r)\cdot b = \frac{e^r-e^{-r}}{2}\cdot a - \frac{e^r+e^{-r}}{2}\cdot b
\end{equation}
With some transformations, we get the roots of $g'(r)$:
\paragraph*{Case $a=b$:}
 \begin{align}
  g'(r) &= 0 \Leftrightarrow\\
  \frac{e^r-e^{-r}}{2}\cdot a &= \frac{e^r+e^{-r}}{2}\cdot a\\
  e^r-e^{-r} &= e^r+e^{-r}\\
  -e^{-r} &= e^{-r}\\
  e^{-r} &= 0\\
 \end{align}
For $a=b$, $d_A$ has no extrema in $\mathbb{R}$.

\paragraph*{ $a\not=b$:}

\begin{align}
g'(r) &= 0 \Leftrightarrow\\
\frac{e^r-e^{-r}}{2}\cdot a &= \frac{e^r+e^{-r}}{2}\cdot b\Leftrightarrow\\
a e^r-ae^{-r} &=  be^r+be^{-r}\Leftrightarrow\\
(a-b)e^r - (a+b)e^{-r} &= 0\Leftrightarrow\\
(a-b)e^r &= (a+b)e^{-r}\Leftrightarrow\\
e^r &= \frac{a+b}{a-b}e^{-r}\Leftrightarrow\\
e^{2r} &= \frac{a+b}{a-b}\Leftrightarrow\\
2r &= \ln\left(\frac{a+b}{a-b}\right)\Leftrightarrow\\
r &= \frac{1}{2}\ln\left(\frac{a+b}{a-b}\right)\label{eq:left-right-extremum}
\end{align}
For  $a\not=b$, $d_A$ has a single extremum at $\frac{1}{2}\ln\left(\frac{a+b}{a-b}\right)$.
This extremum is calculated for both angular boundaries in Lines \ref{line:left-extremum} and \ref{line:right-extremum} of Algorithm~\ref{algo:hyperbolic-distances}.

If $d(r)$ has an extremum $x$ in $A$, the minimum of $d_A(r)$ on $A$ is $\min \{d_A(\text{min}_r)$,  $d_A(\text{max}_r)$, $d_A(x)\}$,
otherwise it is $\min \{d_A(\text{min}_r)$,  $d_A(\text{max}_r)\}$.

\vspace{1\baselineskip}
A similar approach works for the radial boundary curves. Let $B$ be a radial boundary curve at radius $r_B$ and angular bounds $\text{min}_\phi$ and $\text{max}_\phi$.
Let $d_B(\phi)$ be the distance to $q$ restricted to radius $r_B$.
\begin{align}
d_B &: [0,2\pi] \rightarrow \mathbb{R}\\
 d_B(\phi) &= \mathrm{acosh}(\cosh(r_B)\cosh(r_q) - \sinh(r_B)\sinh(r_q) \cos(\phi_q - \phi))
\end{align}
Similarly to the angular boundaries, we define some constants and a function $g(\phi)$ with the same extrema as $d_B$:
\begin{align}
  a &= \cosh(r_B)\cosh(r_q)\\
  b &= \sinh(r_B)\sinh(r_q)\\
 g(\phi) &= a - b \cos(\phi_q - \phi)
\end{align}

\paragraph*{Case: $b = 0$:}
\begin{align}
 b &= \sinh(r_B)\sinh(r_q) = 0 \Leftrightarrow\\
 g(\phi) &= a 
\end{align}
Since $g$ is constant, no extrema exist.

\paragraph*{Case: $b \not= 0$:}
We obtain the extrema with some transformations:
\begin{align}
 g'(\phi) &= -b \sin(\phi_q - \phi)\\
 g'(\phi) &= 0 \Leftrightarrow\\
 \sin(\phi_q - \phi) &= 0 \Leftrightarrow\\
 \phi &= \phi_q \mod \pi
\end{align}
The distance function $d_B(\phi)$ thus has two extrema.

The minimum of $d_B(r)$ on $B$ is then:
\begin{equation}
\min_{r \in B} d_B(r) = \min \{d_B(\text{min}_r), d_B(\text{max}_r)\} \cup \{d_B(\phi) |  \text{min}_\phi \leq \phi \leq \text{max}_\phi \wedge \phi = \phi_q \mod \pi \}
\end{equation}

The distance $\text{dist}_\hyperbolic{} (C, q)$ can thus be written as the minimum of four to ten point-to-point distances. 
Algorithm~\ref{algo:hyperbolic-distances} collects the arguments for these distances in the variable cornerSet and returns the distance minimum as the first return value.
\end{proof}

\begin{lemma}
 Let $T$ be a polar quadtree in Euclidean space, $c$ a quadtree cell of $T$ and $q$ a point in Euclidean space.
 The first value returned by Algorithm~\ref{algo:Euclidean-polar-distances} is the distance of $c$ to $q$.
 \label{lemma:Euclidean-polar-distances}
\end{lemma}

\begin{proof}
The general distance equation for polar coordinates in Euclidean space is
\begin{equation}
 f(r_p, r_q, \phi_p, \phi_q) = \sqrt{r_p^2 + r_q^2 -2 r_p r_q \cos(\phi_p - \phi_q)}
 \label{eq:distance-equation}
\end{equation}

If the query point $q$ is within $C$, the distance is zero.
Otherwise, the distance between $q$ and $C$ is equal to the distance between $q$ and the boundary of $C$.
We consider each boundary component separately and derive the extrema of the distance function.

\paragraph{Radial Boundary.}
When considering the radial boundary, everything but one angle is fixed:
\begin{equation}
  f(\phi_p) = \sqrt{r_p^2 + r_q^2 -2 r_p r_q \cos(\phi_p - \phi_q)}
  \label{eq:distance-only-angle}
\end{equation}
Since the distance is positive and the square root is a monotone function, the extrema of the previous function are at the same values as the extrema of its square $g(\phi)$:
\begin{equation}
  g(\phi_p) =r_p^2 + r_q^2 -2 r_p r_q \cos(\phi_p - \phi_q)
  \label{eq:distance-only-angle-squared}
\end{equation}
We set the derivative to zero to find the extrema:
\begin{align}
  g'(\phi_p) &= 0 \Leftrightarrow\\
  2 r_p r_q \sin(\phi_p - \phi_q)\cdot (\phi_p - \phi_q) &= 0\\
  \phi_p = \phi_q \mod \pi
  \label{eq:distance-only-angle-squared-derivative}
\end{align}

\paragraph{Angular Boundary.}
Similar to the radial boundary, we fix everything but the radius:
\begin{equation}
  f(r_p) = \sqrt{r_p^2 + r_q^2 -2 r_p r_q \cos(\phi_p - \phi_q)}
  \label{eq:distance-only-radius}
\end{equation}

Again, we define a helper function with the same extrema:
\begin{equation}
  g(r_p) =r_p^2 + r_q^2 -2 r_p r_q \cos(\phi_p - \phi_q)
  \label{eq:distance-only-radius-squared}
\end{equation}
We set the derivative to zero to find the extrema:
\begin{align}
  g'(r_p) &= 0 \Leftrightarrow\\
  2r_p - 2r_q\cos(\phi_p - \phi_q) &= 0\Leftrightarrow\\
  r_p &= r_q\cos(\phi_p - \phi_q)\Rightarrow\\
  g(r_p) &= r_p^2 + r_q^2 -2 r_p^2\\ &= r_q^2 - r_p^2\\ &= r_q^2(1-\cos(\phi_p-\phi_q))
  \label{eq:distance-only-angle-radius-derivative}
\end{align}

An extremum of $f$ on the boundary of cell $c$ is either at one of its corners or at the points derived in Eq.~(\ref{eq:distance-only-angle-squared-derivative}) or Eq.~(\ref{eq:distance-only-angle-radius-derivative}).
If $q\not\in c$, the minimum over these points and the corners, as computed by Algorithm~\ref{algo:Euclidean-polar-distances}, is the minimal distance between $q$ and any point in $c$.
If $q$ is contained in $c$, the distance is trivially zero.
\end{proof}
\end{document}

%% file: abbrv.tex
\newcommand{\dist}{\operatorname{dist}}

\newcommand{\ie}{i.\,e.\xspace}

\newcommand{\eg}{e.\,g.\xspace}

\newcommand{\wrt}{w.\,r.\,t.\xspace}

\newcommand{\Pro}[1]{\mathbf{Pr} \left[\,#1\,\right]}